\documentclass[journal,twoside,web]{ieeecolor}
\usepackage{generic}
\usepackage{cite}
\usepackage{amsmath,amssymb,amsfonts}
\usepackage{algorithmic}
\usepackage{graphicx}
\usepackage{algorithm,algorithmic}
\usepackage{hyperref}
\hypersetup{hidelinks=true}
\usepackage{textcomp}

\usepackage{amsmath}
\usepackage{amssymb}
\usepackage{amsfonts}
\usepackage{graphicx}
\usepackage{enumerate}
\usepackage{mathrsfs}
\usepackage{float}
\usepackage{cite}
\usepackage{tikz}
\usetikzlibrary{matrix,arrows,decorations.pathmorphing}
\usepackage{verbatim}
\usepackage{mathtools}
\usepackage{caption}
\usepackage{subcaption}
\usetikzlibrary{arrows}
\usepackage{tabularx}

\newtheorem{cor}{Corollary}

\DeclareMathOperator*{\Dg}{\mathsf{diag}}

\newcommand{\hma}[1]{\textcolor{black}{#1}}

\graphicspath{{./figs/}}

\newtheorem{remark}{Remark}
\newtheorem{lemma}{Lemma}
\newtheorem{thm}{Theorem}

\def\BibTeX{{\rm B\kern-.05em{\sc i\kern-.025em b}\kern-.08em
    T\kern-.1667em\lower.7ex\hbox{E}\kern-.125emX}}
\begin{document}
\title{On the analysis of a higher-order Lotka-Volterra model: an application of $\mathcal{S}$-tensors and the polynomial complementarity problem}
\author{Shaoxuan Cui$^{1,5}$, Qi Zhao$^{2}$, Guofeng Zhang$^{3}$, \IEEEmembership{Member, IEEE}, Hildeberto Jardón-Kojakhmetov$^{1}$ and Ming Cao$^{4}$, 
\thanks{$^{1}$ S. Cui, and H. Jard\'on-Kojakhmetov are with the Bernoulli Institute for Mathematics, Computer Science and Artificial Intelligence, University of Groningen, Groningen, 9747 AG Netherlands {\tt\small \{s.cui, h.jardon.kojakhmetov\}@rug.nl}}
\thanks{$^{2}$ Q. Zhao is with the School of Data Science, Qingdao University of Science and Technology, Qingdao, 266061 China {\tt\small zhaoqi\_1995@163.com}}%
\thanks{$^{4}$ M. Cao is with the Engineering and Technology Institute Groningen, University of Groningen, Groningen, 9747 AG Netherlands {\tt\small m.cao@rug.nl}}
\thanks{$^{5}$ S.Cui was supported by China Scholarship Council.}}

\maketitle

\begin{abstract}
It is known that the effect of species' density on species' growth is non-additive in real ecological systems. 
 This challenges the conventional Lotka-Volterra model, where the interactions are always pairwise and their effects are additive. To address this challenge, we introduce HOIs (Higher-Order Interactions) which are able to capture, for example, the indirect effect of one species on a second one correlating to a third species. Towards this end, we study a general higher-order Lotka-Volterra model. We provide an existence result of a positive equilibrium for a non-homogeneous polynomial equation system with the help of $\mathcal{S}$-tensors. Afterward, by utilizing the latter result, as well as the theory of monotone systems and results from the polynomial complementarity problem, we provide comprehensive results regarding the existence, uniqueness, and stability of the corresponding equilibrium. These results can be regarded as natural extensions of many analogous ones for the classical Lotka-Volterra model, especially in the case of full cooperation, competition among two factions, and pure competition. Finally, illustrative numerical examples are provided to highlight our contributions.
\end{abstract}

\begin{IEEEkeywords}
Tensor, $\mathcal{S}$-tensor, Polynomial complementarity problem, Higher-order Interactions, Lotka-Volterra model, Stability analysis
\end{IEEEkeywords}

\section{Introduction}

\IEEEPARstart{T}he Lotka-Volterra model is one of the most fundamental and widely adopted population models in mathematical biology and ecology, originating from Lotka \cite{lotka1920analytical} and Volterra \cite{volterra1928variations}. An early analysis of the 2-species Lotka-Volterra model was conducted by \cite{goh1976global}. Then, the stability results of the multi-species cooperative model (see Chapter 4 and Definition 16 \cite{sb2010}) were derived by \cite{goh1979stability}, while %the authors of
 \cite{takeuchi1978stability} studied the stability of a generalized multi-species model with both competition and mutualism. Abundant extensive contributions \cite{takeuchi1996global,sb2010,FB-LNS} provide a detailed and comprehensive introduction to the conventional Lotka-Volterra models and their stability results. However, all these conventional models treat the species \emph{pair} as a fundamental unit and only capture \emph{pairwise} interactions, whose effects on the species' growth are additive.
 
Prompted by studies in ecology, such pairwise interaction and its purely additive setting are shown to be insufficient to represent real complex ecological systems, supported, for example, by 
\cite{abrams1983arguments}. Recently \cite{mayfield2017higher}, \emph{HOIs} (Higher-order Interactions) have been introduced to represent non-additive effects and further incorporate some empirical evidence, such as the one showing that HOIs play a significant role in natural plant communities. Followed by the aforementioned idea, %A. D. Letten and D. B. Stouffer 
\cite{letten2019mechanistic} introduced  HOIs into the Lotka–Volterra competition and then demonstrated, by using empirical data and simulations, that HOIs appear under almost all assumptions and help to improve the accuracy of model's predictions. Despite the advantages brought by HOIs, the model becomes mathematically more challenging to analyze. To understand what role HOIs play in influencing the species' coexistence, \cite{singh2021higher}, %utilizing the higher-order Lotka-Volterra model, 
studies the aforementioned problem through simulations. Even more recently, in \cite{gibbs2022coexistence}, numerical simulations with techniques from statistical physics are used to estimate the HOIs' influence on species coexistence. Rigorous mathematical results regarding the existence and stability of equilibria in the higher-order Lotka-Volterra model are still largely missing, mainly because the higher-order system is highly nonlinear.

Since the Lotka-Volterra model is a polynomial system, computing a positive equilibrium is equivalent to solving a system of polynomial equations. With the development of the tensor algebra, a polynomial equation system can be efficiently captured by the tensor-vector product \cite{ding2016solving,wang2019existence,liu2022further}. Researchers have mainly focused on homogeneous polynomial systems and have further shown that such systems may possess a unique positive solution for some particular classes of tensors and under some appropriate conditions. For example, \cite{ding2016solving} provides a result of the existence and uniqueness of a positive solution for $\mathcal{M}$-tensors, while \cite{wang2019existence} extends the results to $\mathcal{H}^+$-tensors, and  \cite{liu2022further} to $\mathcal{S}$-tensors, see section \ref{sec:preliminaries}.  %$\mathcal{S}$-tensors are the most general tensors in this category and we will introduce all these concepts later in this paper. 
However, most of these results are restricted to a homogeneous polynomial case and are not directly applicable to the Lotka-Volterra model.

It is known that the solution set of the linear complementarity problem corresponds to the equilibria of a classical Lotka-Volterra satisfying some extra conditions \cite{takeuchi1996global,takeuchi1980existence}, which are also necessary for the equilibrium to be stable.
%. They further show that these extra conditions are a necessary condition for the equilibrium to be stable \cite{takeuchi1996global}. 
With the introduction of HOIs, the higher-order Lotka-Volterra system is no longer related to the linear complementarity problem but associated with the polynomial complementarity problem introduced in \cite{huang2019tensor,gowda2016polynomial}. The results regarding the polynomial complementarity problem are also a potential tool to analyze the higher-order Lotka-Volterra system.

Alongside developments regarding the stability of the conventional Lotka-Volterra model, there is a long history concerning monotone systems' theory \cite{hirsch2006monotone,smith1988systems}, which is a useful tool whenever the system is cooperative or its Jacobian can be permuted into an irreducible Metzler matrix. Very recently, monotone systems' theory has been applied to study the bi-virus competition model \cite{ye2022convergence}, whereas the tri-virus competition model is shown not to be a monotone system \cite{gracy2022endemic}. In \cite{smith1986competing}, an abstract system of two-subcommunity competition is studied, but the main results are restricted to the positive equilibrium and not applicable to the higher-order system. Furthermore, a generalization of cooperativity is introduced and further analyzed in \cite{kawano2022contraction}. All these tools serve as a good foundation for studying a higher-order Lotka-Volterra model.
 
The contributions of this paper are summarized as follows: firstly, we provide an existence result of a positive equilibrium for a non-homogeneous polynomial system and give a lower and upper bound for the solution of a polynomial complementarity problem under mild conditions, which is an extension of the current results in tensor algebras \cite{ding2016solving,wang2019existence,liu2022further}. Then, inspired by \cite{letten2019mechanistic} and taking both cooperation (mutualism) and competition into account, we propose a general higher-order Lotka-Volterra model. We provide results regarding the existence, uniqueness, and stability of equilibria of the higher-order Lotka-Volterra model. More precise results are given for the case of full cooperation, competition between two factions, and pure competition. Some results are achieved via monotone system theory; others are established by using the properties of tensors and results from the polynomial complementarity problem. To the best of our knowledge, this methodology is novel in the field of the study of dynamical systems. Finally, we provide numerical examples to highlight our theoretical results.

\emph{Notation:} Throughout this paper, $\mathbb{R}$ denotes the set of real numbers, $\mathbb{R}_+$ is the set of nonnegative real numbers, and $\mathbb{R}_{++}$ the set of positive real numbers. For the similar notation $\mathbb R^n$ ($\mathbb{R}^n_+$ or $\mathbb{R}^n_{++}$), the superscript stands for the dimension of the space. Whenever $a,b\in \mathbb{R}^n$, we use the notation $a \geq (\leq) b$ to denote that $a_i \geq (\leq) b_i$, for all $i=1,\ldots,n$; $a>  (<) b$ to denote that $a_i >(<) b_i$, for all $i=1,\ldots,n$. This element-wise comparison is also used for matrices and tensors. 
For simplicity, the equilibrium point denotes both the equilibrium point itself and the vector constructed from the equilibrium point. For a vector $X=(x_1,\ldots,x_n)\in\mathbb R^n$, the vector's $L_{\infty}$ norm is defined as $\|X\|_{\infty}=\max _{i}\left|x_i\right|$.

\section{Tensor-based approach}\label{sec:preliminaries}

In this section, we briefly introduce the concepts of tensors and polynomial complementarity problems that are useful later in the paper. We also propose some novel results, which are extensions of some known results in tensor algebra \cite{ding2016solving,wang2019existence,liu2022further}. We will then show that all these results are very useful tools to analyze the higher-order Lotka-Volterra system.

\subsection{Tensor}
A tensor $T\in\mathbb{R}^{n_1\times n_2 \times \cdots \times n_k}$ is a multidimensional array, where the order is the number of its dimensions $k$ and each dimension $n_i$, $i=1,\cdots,k$ is a mode of the tensor. A tensor is cubical if every mode has the same size, that is $T\in\mathbb{R}^{n\times n \times \cdots \times n}$. From now on we denote a $k$-th order $n$ dimensional cubical tensor as $T\in\mathbb{R}^{n\times n \times \cdots \times n}=\mathbb{R}^{[k,n]}$. Throughout this paper, a `tensor' always refers to a cubical tensor. A tensor $T$ is called supersymmetric if $T_{j_1 j_2 \ldots j_k}$ is invariant under any permutation of the indices.
The identity tensor $\mathcal{I}=\left(\delta_{i_1 \cdots i_m}\right)$ is defined by
\begin{equation*}
   \delta_{i_1 \cdots i_m}= \begin{cases}1 & \text { if } i_1=i_2=\cdots=i_m, \\ 0 & \text { otherwise.}\end{cases}
\end{equation*}

We then consider the following notation:
 $A x^{m-1}$ and $x^{[m-1]}$ are vectors, whose $i$-th components are
$$
\begin{aligned}
\left(A x^{m-1}\right)_i & =\sum_{i_2, \ldots, i_m=1}^n A_{i, i_2 \cdots i_m} x_{i_2} \cdots x_{i_m}, \\
\left(x^{[m-1]}\right)_i & =x_i^{m-1};
\end{aligned}
$$
where $m$ is the order of the corresponding tensor $A$.
For a tensor, consider the following eigenvalue eigenvector problem:
\begin{equation}\label{eq:eigenproblem}
A x^{m-1}=\lambda x^{[m-1]},
\end{equation}
where if there is a real number $\lambda$ and a nonzero real vector $x$ that are solutions of \eqref{eq:eigenproblem}, then $\lambda$ is called an H-eigenvalue of $A$ and $x$ is the H-eigenvector of $A$ associated with $\lambda$. Throughout this paper, the words eigenvalue and eigenvector as well as H-eigenvalue and H-eigenvector are used interchangeably.
The spectral radius of the tensor $A$ is
$
\rho(A)=\max \{|\lambda|: \lambda \text { is an eigenvalue of } A\}.
$

\subsection{$\mathcal{M}$-tensors and $\mathcal{H}^{+}$-tensors}
Here, we briefly recall the concepts of $\mathcal{M}$-tensor and $\mathcal{H}^{+}$-tensor introduced in \cite{wang2019existence,ding2013m,zhang2014m,cui2024metzler}.
A tensor $A \in \mathbb{R}^{[m, n]}$  is an $\mathcal{M}$-tensor \cite{ding2013m,zhang2014m} if it can be represented as $A=s \mathcal{I}-B$, where $\mathcal I\in\mathbb R^{[m,n]}$ is the identity tensor, $B\in\mathbb R^{[m,n]}$ is a nonnegative tensor (i.e., each entry of $B$ is nonnegative), and $s \geq \rho(B)$. Furthermore, $A$ is called a nonsingular $\mathcal{M}$-tensor if $s>\rho(B)$. A tensor $A \in \mathbb{R}^{[m, n]}$ is a Metzler tensor \cite{cui2024metzler} if it can be represented as $A=s \mathcal{I}+B$ where $s\in \mathbb{R}$ and $B\in\mathbb R^{[m,n]}$ is a nonnegative tensor. In other words, a Metzler tensor is a tensor whose off-diagonal elements are all nonnegative.

Let $A=\left(A_{i_1 i_2 \ldots i_m}\right)\in\mathbb{R}^{[m, n]}$. We define the tensor $\langle A\rangle=\left(m_{i_1 i_2 \ldots i_m}\right)$ as \emph{the comparison tensor of $A$} by
$$
m_{i_1 i_2 \ldots i_m}=\left\{\begin{array}{cl}
\left|A_{i_1 i_2 \ldots i_m}\right|, & \left(i_1 i_2 \ldots i_m\right)=\left(i_1 i_1 \ldots i_1\right), \\
-\left|A_{i_1 i_2 \ldots i_m}\right|, & \left(i_1 i_2 \ldots i_m\right) \neq\left(i_1 i_1 \ldots i_1\right) .
\end{array}\right.
$$
Notice that a tensor $A\in\mathbb{R}^{[m, n]}$ is an $\mathcal{H}$-tensor \cite{wang2019existence} if its comparison tensor is an $\mathcal{M}$-tensor; it is a nonsingular $\mathcal{H}$-tensor, if its comparison tensor is a nonsingular $\mathcal{M}$-tensor. A nonsingular $\mathcal{H}$-tensor $A$ with all its diagonal elements $a_{i i \ldots i}>0$ is called an $\mathcal{H}^{+}$-tensor \cite{wang2019existence}. 

A tensor ${A}\in\mathbb{R}^{[m, n]}$ is called diagonally dominant \cite{qi2005eigenvalues} if
$$
\left|A_{i i \ldots i}\right| \geq \sum_{\left(i_2, \ldots, i_m\right) \neq(i, \ldots, i)}\left|A_{i i_2 \ldots i_m}\right|
$$
for all $i=1,2, \ldots, n$; and is called strictly diagonally dominant if
$$
\left|A_{i i \ldots i}\right|>\sum_{\left(i_2, \ldots, i_m\right) \neq(i, \ldots, i)}\left|A_{i i_2 \ldots i_m}\right| 
$$
for all $i=1,2, \ldots, n$.
Nonsingular $\mathcal{M}$-tensors and strictly diagonally dominant tensors with positive diagonal elements are two special types of $\mathcal{H}^{+}$-tensors. These have the following useful properties:

%\hjk{homogenize the notation $\mathbf{0}$ or $0$ for the zero vector throughout the text.}

\begin{lemma}[\cite{wang2019existence}]
If $A\in\mathbb{R}^{[m, n]}$ is an $\mathcal{H}^{+}$-tensor and $\langle A\rangle$ is its comparison tensor, then there exists a positive vector $x>\mathbf{0}$ such that $A x^{m-1}>\mathbf{0}$ and $\langle A\rangle x^{m-1}>\mathbf{0}$.
\end{lemma}

\subsection{$\mathcal{S}$-tensors and their properties}
Here, we briefly recall the concept of an $\mathcal{S}$-tensor and we leverage their properties \cite{liu2022further} to achieve some further results regarding non-homogeneous systems related to $\mathcal{S}$-tensors.

A tensor $A\in \mathbb{R}^{[m, n]}$ is an $\mathcal{S}$-tensor \cite{liu2022further}, if there exists a positive vector $x$ such that $Ax^{k-1}>\mathbf{0}$. Clearly, an $\mathcal{H}^{+}$-tensor is an $\mathcal{S}$-tensor.

The following lemmas tell us that, in particular, if $A$ is an $\mathcal H^+$-tensor, then the tensor equation $Ax^{k-1}=b$ has a unique positive solution for every positive $b$.

\begin{lemma}[\cite{liu2022further}]\label{lem:stensor}
    If a tensor $A\in\mathbb{R}^{[m, n]}$ satisfies: i) the set $D_{A}=\left\{x \in \mathbb{R}_{++}^n: A x^{m-1}>0\right\}$ is nonempty (i.e. $A$ is an $\mathcal{S}$-tensor) and ii) the map $\eta \mathcal{I}-A: \mathbb{R}_{++}^n \rightarrow \mathbb{R}_{++}^n$, defined by $x \mapsto(\eta \mathcal{I}-A) x^{m-1}$, is an increasing map on the set $D_{A}$ for some $\eta>0$, then for every positive vector $b\in\mathbb R^n$ the tensor equation $A x^{m-1}=b$ has a unique positive solution.
\end{lemma}

Notice indeed, that if $A$ is an $\mathcal{H}^{+}$-tensor or an $\mathcal{M}$-tensor, then both conditions in Lemma \ref{lem:stensor} hold automatically, see \cite[section 3]{wang2019existence}. Then, we have the following.

\begin{lemma}[\cite{wang2019existence}]\label{lem:h+tensor}
    If $A\in\mathbb{R}^{[m, n]}$ is an $\mathcal{H}^{+}$-tensor, then for every positive vector $b\in \mathbb R_{++}^n$, the tensor equation $A x^{m-1}=b$ has a unique positive solution.
\end{lemma}

However, we are interested in the non-homogeneous equation system
\begin{equation}\label{eq:nhte}
    A_{k}x^{k-1}+A_{k-1} x^{k-2}+\cdots  A_{2} x=b,
\end{equation}
where $b$ is a positive vector. We will show later that this equation system is relevant to the higher-order Lotka-Volterra model. The following Lemma will thus be important.

%We will use the following Lemma to provide the result.

\begin{lemma}[Banach fixed point theorem, Lemma 3.1 \cite{wang2019existence}]\label{lem:banach}
    Let $\mathbb{P}$ be a regular cone in an ordered Banach space $\mathbb{E}$ and $[u, v] \subset \mathbb{E}$ be a bounded ordererd interval. Suppose that $T:[u, v] \rightarrow \mathbb{E}$ is an increasing continuous map which satisfies
$$
u \leq T(u) \quad \text { and } \quad v \geq T(v).
$$
Then $T$ has at least one fixed point in $[u, v]$. Moreover, there exists a minimal fixed point $x_*$ and a maximal fixed point $x^*$ in the sense that every fixed point $\bar{x}$ satisfies $x_* \leq \bar{x} \leq x^*$, with $x_*$ and $x^*$ also in the interval $[u,v]$. Finally, the sequence defined by
%iteration method
$$
x_{k+1}=T\left(x_k\right), \quad k=0,1,2, \ldots,
$$
converges to $x_*$ from below if $x_0=u$, i.e.,
$$
u=x_0 \leq x_1 \leq \cdots \leq x_*,
$$
and converges to $x^*$ from above if $x_0=v$, i.e.,
$$
v=x_0 \geq x_1 \geq \cdots \geq x^*.
$$
\end{lemma}

We now show that, under certain conditions, the non-homogeneous tensor equation \eqref{eq:nhte} has a unique solution.

\begin{thm}\label{thm:st}
    The non-homogeneous tensor equation \eqref{eq:nhte} with $b=\mathbf{1}$ has a unique positive solution if 
    $A_{k-1},A_{k-2},\cdots,A_2$ are all $\mathcal{S}$-tensors associated with the same positive vector $v$, i.e., $A_{k-1}v^{k-1}>0$, $A_{k-1}v^{k-2}>0,\cdots , A_2v>0$.
\end{thm}

\begin{proof}
    %The proof is analog to the idea in \cite{liu2022further}.
%
   %Followed from \cite[lemma 2.3]{wang2019existence}, any $\mathcal{H}^+$-tensor $A_i$ can be formulated as $A_i=\mathcal{W}_i +\mathcal{N}_i=\alpha_i \mathcal{I}-\mathcal{B}_i+\mathcal{N}_i$,
  % where $\mathcal{W}_i=\alpha_i \mathcal{I}-\mathcal{B}_i$ is a $\mathcal{M}-$tensor which includes diagonal elements and nonpositive elements of $A_i$, and $\mathcal{N}_i$ is an nonnegative tensor which includes positive elements of $A_i$ other than diagonal elements. %From \cite[Section 3]{wang2019existence}, we know that $(\eta_i \mathcal{I} +\mathcal{B}_i-\mathcal{N}_i)x^{i-1}$ is an increasing map for some suffieciently large $\eta_i >0$ on $\mathbb{R}^n_{++}$.
%
Let $A_i\in\mathbb R^{[i,n]}$.
   Following \cite{liu2022further},  any tensor $A_i$ can be rewritten as $A_i=\alpha_i \mathcal{I}-\mathcal{B}_i+\mathcal{N}_i$,
 where $\alpha_i=\max \left\{0, (A_i)_{1 \cdots 1}, (A_i)_{2 \cdots 2}, \ldots, (A_i)_{n \cdots n}\right\}, \mathcal{B}_i$ is a nonnegative tensor whose diagonal entries $(\mathcal{B}_i)_{j j \cdots j}=\alpha-$ $(A_i)_{j j \ldots j}\;(j=1,2, \ldots, n)$ and nondiagonal entries equal to the opposite of the negative elements of $A_i$ other than diagonal elements, and $\mathcal{N}_i$ is a nonnegative tensor which includes the positive elements of $A_i$ other than diagonal elements.

   Next, let $f(x)=\sum_{i=2}^{k}(\alpha_i \mathcal{I}+\mathcal{N}_i)x^{i-1}$. We further define $g(x)= \sum_{i=2}^{k} \mathcal{B}_i x^{i-1}+b$. Now, since $b$ is a positive vector and $\mathcal{B}_i$ is nonnegative, we have that $g(x)$ is an increasing map on $\mathbb{R}^n_{++}$. 
   The map $f(x): \mathbb{R}^n_{++} \rightarrow \mathbb{R}^n_{++}$ is an increasing continuous map and thus it has an increasing continuous inverse $f^{-1}(x)$ on $\mathbb{R}^n_{++}$. Hence, a solution of the tensor equation \eqref{eq:nhte} can be written as:
\begin{equation*}
x=T(x)=f^{-1}(g(x)).
\end{equation*}
Thus, the solution of the tensor equation \eqref{eq:nhte} corresponds to a fixed point of $T(x)=f^{-1}(g(x))$. The map $T(x)$ is an increasing map in $\mathbb{R}^n_{++}$.
We notice that since $A_i v^{i-1}>\mathbf{0}$, we have $A_i (wv)^{i-1}>\mathbf{0}$ for any positive scalar $w>0$. Thus, we can choose a $w>0$ such that $\sum_{i=2}^k A_i (wv)^{i-1}$ is sufficiently large (or small).
 Now we define 
\begin{equation}\label{eq:bargamma1}
\begin{split}
\Bar{\gamma}&=\max _{i=1,2, \ldots, n} \frac{1}{\left(\sum_{i=2}^k A_i (wv)^{i-1}\right)_i}\\
&=\frac{1}{\min _{i=1,2, \ldots, n}\left(\sum_{i=2}^k A_i (wv)^{i-1}\right)_i}.
\end{split}
\end{equation}

Here, we choose a $w$ such that $\Bar{\gamma}=1$. Let $\Bar{x}=wv$. We have
\begin{equation}
\begin{aligned}
g(\bar{x}) & =\sum_i (\mathcal{B}_i) (w v)^{i-1}+b \\
& \leq  \sum_i (\mathcal{B}_i)(w v)^{i-1}+  \sum_i A_i(w v)^{i-1} \\
& =  \sum_i (\alpha_i\mathcal{I} +\mathcal{N}_i) (w v)^{i-1}=f(\bar{x}).
\end{aligned}
\end{equation}
This shows that $T(\Bar{x})=f^{-1}(g(\Bar{x}))\leq \Bar{x}$.

Similarly, let us define 
\begin{equation}\label{eq:undergamma}
\begin{split}
\underline{\gamma}&=\min _{i=1,2, \ldots, n} \frac{1}{\left(\sum_{i=2}^k A_i (tv)^{i-1}\right)_i}\\
&=\frac{1}{\max _{i=1,2, \ldots, n}\left(\sum_{i=2}^k A_i (tv)^{i-1}\right)_i}.
\end{split}
\end{equation}
Here, we choose a $t$ such that $\underline{\gamma}=1$. Letting $\underline{x}=tv$ 
we notice that $t\leq w$ implies that $\underline{x}\leq \Bar{x}$. Analogously,
\begin{equation}
\begin{aligned}
g(\underline{x}) & =\sum_i (\mathcal{B}_i) (tv)^{i-1}+b \\
& \geq  \sum_i (\mathcal{B}_i)(tv)^{i-1}+  \sum_i A_i(tv)^{i-1} \\
& =  \sum_i (\alpha_i\mathcal{I} +\mathcal{N}_i )(tv)^{i-1}=f(\underline{x}).
\end{aligned}
\end{equation}
This shows that $T(\underline{x})=f^{-1}(g(\underline{x}))\geq \underline{x}$. By Lemma \ref{lem:banach}, there exists at least one positive solution.

Next, we prove the uniqueness by contradiction. Suppose there are 2 fixed points $T(x)=x>0, T(y)=y>0$. Define $\tau=\min \frac{x_i}{y_i}$ by using their $i$-th component, then $x \geq \tau y $ and $ x_j=\tau y_j$ for some $j$. 
If $\tau<1$, we have 
\begin{equation}
\begin{aligned}
g(\tau y) & =\sum_i (\mathcal{B}_i) (\tau y)^{i-1}+b \\
& =\sum_i (\mathcal{B}_i)(\tau y)^{i-1}+  \sum_i A_i y^{i-1} \\
& >\sum_i (\mathcal{B}_i)(\tau y)^{i-1}+  \sum_i A_i (\tau y)^{i-1} \\
& =  \sum_i (\alpha_i\mathcal{I}+\mathcal{N}_i)  (\tau y)^{i-1}=f(\tau y).
\end{aligned}
\end{equation}

This leads to $\tau {y_j}<[T(\tau y)]_j \leqslant[T(x)]_j=x_j=\tau {y_j}$, which is a contradiction. Thus, it must hold $\tau\geq 1$, which leads to $x\geq y$.  So, let $\xi=\max \frac{x_i}{y_i}$. Then, $x \leqslant \xi y$ and $ x_j=\xi y_j$ for some $j$. If $\xi>1$, we have 
\begin{equation}
\begin{aligned}
g(\xi y) & =\sum_i (\mathcal{B}_i) (\xi y)^{i-1}+b \\
& =\sum_i (\mathcal{B}_i)(\xi y)^{i-1}+  \sum_i A_i y^{i-1} \\
& <\sum_i (\mathcal{B}_i)(\xi y)^{i-1}+  \sum_i A_i (\xi y)^{i-1} \\
& =  \sum_i (\alpha_i\mathcal{I} +\mathcal{N}_i) (\xi y)^{i-1}=f(\xi y).
\end{aligned}
\end{equation}
Then, $\xi {y_j}>[T(\xi y)]_j \geq[T(x)]_j=x_j=\xi {y_j}$. This is again a contradiction. Thus, it must hold $\xi\leq 1$, which leads to $y\geq x$. Combined with the above argument $x\geq y$, we have $x=y$. Therefore the fixed point is unique.
\end{proof}

\begin{remark}
    Theorem \ref{thm:st} requires $b=\mathbf{1}$. However, the result is still very general. Consider \eqref{eq:nhte} with an arbitrary positive $b$ and any arbitrary tensors $A_i$. Let $(\Tilde{A}_i)_{j_1,\cdots,j_i}=\frac{(A_i)_{j_1,\cdots,j_i}}{b_{j_1}}$. 
    Then, \eqref{eq:nhte} is equivalent to $\Tilde{A}_{k-1}x^{k-1}+\Tilde{A}_{k-1} x^{k-2}+\cdots  \Tilde{A}_{2} x=\mathbf{1}$. Furthermore, the way we construct $f(x), g(x)$ for the proof of Theorem \ref{thm:st} is different from the proof of Lemma \ref{lem:stensor}. Our approach can deal with a more complicated non-homogeneous case and doesn't rely on condition ii) in Lemma \ref{lem:stensor}. As a special case of Theorem \ref{thm:st}, the following Corollary further generalizes the result of Lemma \ref{lem:stensor}.
\end{remark}

\begin{cor}\label{cor:st}
    If $A$ is an $\mathcal{S}$-tensor, then the homogeneous tensor equation $Ax^{k-1}=b$ with $b=\mathbf{1}$ has a unique positive solution.
\end{cor}

Next, we present a result similar to Theorem \ref{thm:st} but for  $\mathcal{M}$-tensors.

\begin{thm}\label{thm:mt}
    The non-homogeneous tensor equation \eqref{eq:nhte} has at least one nonnegative solution if 
    $A_{k-1},A_{k-2},\cdots,A_2$ are all $\mathcal{M}$-tensors with the form $A_i=\alpha_i \mathcal{I}- \mathcal{B}_i$, and are all associated with the same positive vector $v$, i.e., $A_{k-1}v^{k-1}>0$, $A_{k-1}v^{k-2}>0,\cdots , A_2v>0$, where $\mathcal{B}_i$ is a nonnegative tensor, and $\alpha_i$ a positive scalar.  Furthermore, let $f(x)=\sum_{i=2}^{k}(\alpha_i \mathcal{I})x^{i-1}$ and $g(x)=\sum_{i=2}^{k}\mathcal{B}_i x^{i-1}+b$. They implicitly define the map $T(x)=f^{-1}(g(x))$. If the iteration $x_{k+1}=T\left(x_k\right), \quad k=0,1,2, \ldots$ converges to $x_*>0$ from the initial condition $x_0=0$, then the non-homogeneous tensor equation \eqref{eq:nhte} has at least one positive solution and there is no solution with zero entry (boundary solution); if furthermore the iteration $x_{k+1}=T\left(x_k\right), \quad k=0,1,2, \ldots$ also converges to $x_*$ from the initial condition $x_0=wv$, where $w$ is a positive scalar such that $0<\bar{\gamma}=\max _i \frac{b_i}{A_k(w v)^{k-1}+A_{k-1}(w v)^{k-2}+\cdots+A_2(w v)}\leq 1$ holds, then the positive solution is unique. 
\end{thm}
\begin{proof}
The solution of the tensor equation \eqref{eq:nhte} can be rewritten as $\sum_{i=2}^{k}(\alpha_i \mathcal{I})x^{i-1}=\sum_{i=2}^{k}\mathcal{B}_i x^{i-1}+b$, which from the definitions made in the statement yields $f(x)=g(x)$. Since $\mathcal{B}_i$ is a nonnegative tensor, $g(x)$ is an increasing map on $\{x|x\geq \mathbf{0}\}$.
Since the function $f(x)$ is an increasing continuous map on $\{x|x\geq \mathbf{0}\}$ Then, it has an increasing continuous inverse $f^{-1}(x)$. Hence, the tensor equation \eqref{eq:nhte} can be written as:
\begin{equation*}
x=T(x)=f^{-1}(g(x)).
\end{equation*}
The solution of the tensor equation \eqref{eq:nhte} corresponds to a fixed point of $T(x)=f^{-1}(g(x))$. Moreover, $T(x)$ is an increasing map on $\{x|x\geq \mathbf{0}\}$. So, the rest of the proof consists on showing that $T$ is a self-map.

We notice that since $A_i v^{i-1}>\mathbf{0}$, we have $A_i (wv)^{i-1}>\mathbf{0}$ for any positive scalar $w>0$. Thus, we can choose a $w>0$ such that $\sum_{i=2}^k A_i (wv)^{i-1}$ is sufficiently large (or small). Now we define 
\begin{equation}\label{eq:bargamma}
\Bar{\gamma}=\max _{i=1,2, \ldots, n} \frac{b_i}{\left(\sum_{i=2}^k A_i (wv)^{i-1}\right)_i}.
\end{equation}
Here, we choose a $w$ such that $0<\Bar{\gamma}\leq 1$. Moreover, it holds $\Bar{\gamma}^p\geq \Bar{\gamma}$ with $0\leq p \leq 1$. Let $\bar{x}=\bar{r}^{\frac{1}{k-1}} w v$. We have

\begin{equation}
\begin{aligned}
g(\bar{x}) & =\sum_i \mathcal{B}_i \bar{\gamma}^{\frac{i-1}{k-1}}(w v)^{i-1}+b \\
& \leq \sum_i \mathcal{B}_i \bar{\gamma}^{\frac{i-1}{k-1}}(w v)^{i-1}+  \bar{\gamma} \sum_i A_i(w v)^{i-1} \\
& \leq  \sum_i (\mathcal{B}_i+A_i)  \bar{\gamma}^{\frac{i-1}{k-1}}(w v)^{i-1} \\
&= \sum_i \alpha_i  \bar{\gamma}^{\frac{i-1}{k-1}}(w v)^{i-1}=f(\bar{x}).
\end{aligned}
\end{equation}

This leads to $T(\bar{x})\leq \bar{x}$. Next, we further look at
\begin{equation}
g(\mathbf{0})  =b =f(y)= \sum_{i=2}^{k}(\alpha_i \mathcal{I})y^{i-1}.
\end{equation}
Since $b$ is positive, $y$ must also be positive. This further implies that $T(\mathbf{0})=f^{-1}(b)=y\geq \mathbf{0}$. From Lemma \ref{lem:banach}, there must be at least one positive solution. We know from Lemma \ref{lem:banach} that the iteration must converge to the boundary solution from below and if $x_* = x^*$, then the positive solution is unique.
\end{proof}

\begin{remark}
    According to Lemma \ref{lem:stensor}, there are usually multiple positive vectors $y$ such that $Ay^{k-1}>0$. Thus, it is not unreasonable to expect that $A_i$ are all $\mathcal{M}$- or $\mathcal{H}^+$-tensors associated with the same positive vector.
\end{remark}

\subsection{Polynomial complementarity problem}
In tensor algebra, the polynomial complementarity problem is strongly related to the properties of tensors \cite{huang2019tensor}. The polynomial complementarity problem (PCP) is defined by
\begin{equation}\label{eq:pcp}
    \begin{split}
        x &\geq \mathbf{0}, \\
        F(x)+q=\sum_{k=2}^m {A}_k x^{k-1}+q &\geq \mathbf{0},  \\ x^{\top}(F(x)+q)=x^{\top}\left(\sum_{k=2}^m {A}_k x^{k-1}+q\right)&=0,
    \end{split}
\end{equation}
% {\small
% \begin{equation}\label{eq:pcp}
% x \geq 0, \quad \sum_{k=2}^m {A}_k x^{k-1}+q \geq 0, \quad \text { and } \quad x^{\top}\left(\sum_{k=2}^m {A}_k x^{k-1}+q\right)=0,
% \end{equation}}
for any given $\left({A}_2, {A}_3, \ldots, {A}_m\right) \in \mathbb{R}^{[2, n]} \times \mathbb{R}^{[3, n]} \times \cdots \times \mathbb{R}^{[m, n]}$, any real vector $q$ and a given function $F(x):=\sum_{k=2}^m {A}_k x^{k-1}$.  
Let $F^{\infty}(x):=\lim _{\lambda \rightarrow \infty} \frac{F(\lambda x)}{\lambda^{m-1}}={A}_m x^{m-1}$, and $\operatorname{SOL}(F, q)$ denote the solution set of the problem \eqref{eq:pcp} for a given function $F$ and parameter $q$. The following are some important results of the polynomial complementarity problem. It is straightforward to see that a positive solution of $F(x)+q=\mathbf{0}$ must be a solution of the corresponding PCP.
The first lemma shows when the problem has a bounded solution set.

\begin{lemma}[Theorem 7.1 \cite{huang2019tensor}, Proposition 2.1 \cite{gowda2016polynomial}]\label{lem:pcp}
    For the polynomial $F$, consider the statements:
\begin{itemize}
    \item[(a)] $\operatorname{SOL}\left(F^{\infty}, 0\right)=\{0\}$.
    \item[(b)] For any bounded set $\Omega \subset \mathbb{R}^n, \bigcup_{q \in \Omega} \operatorname{SOL}(F, q)$ is bounded.
\end{itemize}
Then, (a) implies (b). The reverse implication holds when $F$ is homogeneous.
\end{lemma}

The next lemma further indicates when the solution set is nonempty.
\begin{lemma}[Theorem 7.2 \cite{huang2019tensor}, Theorem 5.1 \cite{gowda2016polynomial}]\label{lem:pcp2}
    Let $F: \mathbb{R}^n \rightarrow \mathbb{R}^n$ be a polynomial mapping with leading term $F^{\infty}$. Suppose that there is a $d \in \mathbb{R}_{+}^n$ such that one of the following conditions holds:
    \begin{itemize}
        \item[(a)] $\operatorname{SOL}\left(F^{\infty}, 0\right)=\{0\}=\operatorname{SOL}\left(F^{\infty}, d\right)$.
        \item[(b)] $\operatorname{SOL}\left(F^{\infty}, 0\right)=\{0\}=\operatorname{SOL}(F, d)$.
    \end{itemize}

Then, for all $q \in \mathbb{R}^n$, the $\operatorname{PCP}(F, q)$ has a nonempty and compact solution set.
\end{lemma}

The polynomial complementarity problem is said to have the property of global uniqueness and solvability (denoted by GUS-property) if and only if it has a unique solution for every $q \in \mathbb{R}^n$. A mapping $F: \mathbb{S} \subseteq \mathbb{R}^n \rightarrow \mathbb{R}^n$ is said to be a $P$-function on $\mathbb{S}$, if and only if $\max _{i =1,\cdots,n}\left(x_i-y_i\right)\left[F_i(x)-F_i(y)\right]> 0$ for all $x, y \in \mathbb{S}$ with $x \neq y$. The following lemma describes when the PCP has the GUS-property.

\begin{lemma}[Theorem 7.4 \cite{huang2019tensor}, Theorem 6.1 \cite{gowda2016polynomial}]\label{lem:pcp3}
Let $F: \mathbb{R}^n \rightarrow \mathbb{R}^n$ be a polynomial mapping such that $\operatorname{SOL}\left(F^{\infty}, 0\right)=\{0\}$. Then, the following are equivalent:\\
(a) $\operatorname{PCP}(F, q)$ has the GUS-property.\\
(b) $\operatorname{PCP}(F, q)$ has at most one solution for every $q \in \mathbb{R}^n$.

Moreover, condition (b) holds when $F$ is a $P$-function on $\mathbb{R}_{+}^n$.
\end{lemma}

Now, equipped with the preliminaries introduced above, we consider the special case where $F(x)=Ax+Bx^2$. The polynomial complementarity problem now becomes a quadratic complementarity problem as follows:
% {\small
% \begin{equation}\label{eq:qcp}
% x \geq 0, \quad  Bx^2+Ax+q \geq 0, \quad \text { and } \quad x^{\top}\left( Bx^2+Ax+q\right)=0.
% \end{equation}}
\begin{equation}\label{eq:qcp}
\begin{split}
    x \geq \mathbf{0}\\
    Bx^2+Ax+q \geq \mathbf{0}\\
    x^{\top}\left( Bx^2+Ax+q\right)=0.
\end{split}
\end{equation}

Next, we want to derive the lower and upper bound of the quadratic complementarity problem \eqref{eq:qcp} under some mild conditions. Firstly, we recall some concepts introduced in \cite{xu2019estimations}. %\hjk{You may want to improve the notation, this does not look good}\csx{actually I copied them from the original paper}

Let ${A}=\left(a_{i_1 i_2 \cdots i_m}\right) \in \mathbb{R}^{[m, n]}$. Define
\begin{equation}\label{eq:r+}
r_i({A})_{+}:=\sum_{a_{i i_2 i_3 \cdots i_m} \geq 0,\left(i_2, i_3, \ldots, i_m\right) \neq(i, i, \ldots, i)} a_{i i_2 i_3 \cdots i_m},
\end{equation}
and
\begin{equation}\label{eq:r-}
r_i({A})_{-}:=\sum_{a_{i i_2 i_3 \cdots i_m}<0,\left(i_2, i_3, \ldots, i_m\right) \neq(i, i, \ldots, i)}\left|a_{i i_2 i_3 \cdots i_m}\right|.
\end{equation}
When there is no $a_{i i_2 i_3 \cdots i_m} \geq 0$ for all $\left(i_2, i_3, \ldots, i_m\right) \neq(i, i, \ldots, i)$, or $a_{i i_2 i_3 \cdots i_m}<0$ for all $\left(i_2, i_3, \ldots, i_m\right) \neq(i, i, \ldots, i)$, then $r_i({A})_{+}:=0$ and $r_i({A})_{-}:=0$ respectively.

The tensor ${A}=\left(a_{i_1 i_2 \cdots i_m}\right) \in \mathbb{R}^{[m, n]}$ is said to be a \emph{generalized row strictly diagonally dominant} tensor if and only if for all $i \in[n]$,
$
\left|a_{i i \cdots i}\right|-r_i({A})_{-}>0,
$
where $r_i({A})_{-}$is defined by \eqref{eq:r-}. If ${A}$ additionally satisfies $a_{i i \cdots i}>0$ for all $i \in[n]$, then we call ${A}$ a generalized row strictly diagonally dominant tensor with all positive diagonal entries.

\begin{remark}
We now know that a strictly diagonally dominant tensor with positive diagonal elements is a generalized row strictly diagonally dominant tensor with all positive diagonal entries and an $\mathcal{S}-$tensor as well as an $\mathcal{H}^+-$tensor. All the following results regarding a generalized row strictly diagonally dominant tensor with all positive diagonal entries is also applicable to a strictly diagonally dominant tensor with positive diagonal elements.
\end{remark}

%\hjk{It may make sense to use abbreviations...}

Now, let \begin{equation}\label{eq:omega}
\Omega(q):=\left\{i \in[n]: q_i<0\right\}, \quad \forall q \in \mathbb{R}^n \backslash \mathbb{R}_{+}^n.
\end{equation}

The set $\Omega(q)$ refers to the set of indices where the $i$-th component of $q$ is strictly negative.

\begin{lemma}\label{lem:qcp1}
   Let $B \in \mathbb{R}^{[3, n]}$ and $A\in\mathbb{R}^{[2, n]}$ be generalized row strictly diagonally dominant tensors with all positive diagonal entries, $q \in \mathbb{R}^n \backslash \mathbb{R}_{+}^n$ be any given vector, and $\Omega(q)$ be defined by \eqref{eq:omega}. If $x \in \mathbb{R}^n$ is a solution of \eqref{eq:qcp} with $\|x\|_{\infty}=x_k$ for some $k$, then $k \in \Omega(q)$ and
$$
\left(Bx^2+Ax\right)_k+q_k=\mathbf{0}.
$$
\end{lemma}

\begin{proof}
The proof is similar to \cite{xu2019estimations}.
    If $k \notin \Omega(q)$, then $q_k \geq 0$. Since 0 is not a solution of \eqref{eq:qcp} when $q \notin \mathbb{R}_{+}^n$, it follows that $x_k=\|x\|_{\infty}>0$. Thus, we have
\begin{equation}
\begin{aligned}
&\frac{\left(Bx^2+Ax+q\right)_k}{\|x\|_{\infty}^{2}} =\frac{\left(Bx^2+Ax\right)_k}{\|x\|_{\infty}^{2}}+\frac{q_k}{\|x\|_{\infty}^{2}} \\
& \geq \frac{\left(Bx^2+Ax\right)_k}{\|x\|_{\infty}^{2}} 
 =\sum_{i_2, i_3}^n B_{k i_2 i_3} \frac{x_{i_2} x_{i_3}}{\|x\|_{\infty}^{2}} + \sum_{i_2} A_{k i_2} \frac{x_{i_2}}{\|x\|_{\infty}^2}\\
& \geq B_{k k k}-r_k(B)_{-} + \frac{A_{k k}-r_k(A)_{-}}{\|x\|_{\infty}}\\
& >0 .
\end{aligned}
\end{equation}
Then we have
$
x_k\left(Bx^2+Ax+q\right)_k>0,
$
which contradicts the fact that $x$ solves \eqref{eq:qcp}. Thus, it holds that $k \in$ $\Omega(q)$. Moreover, we have $\left(Bx^2+Ax\right)_k+q_k=0$ because $x_k>0$ and $x_k\left(Bx^2+Ax+q\right)_k=0$.
\end{proof}

Let $s(B)_k=B_{k k k}+r_k(B)_{+}$ and $s(A)_k=A_{kk}+r_k(A)_{+}$. Furthermore, define $\delta(B)_k= B_{k k k}-r_k(B)_{-}$ and $\delta(A)_k= A_{k k}-r_k(A)_{-}$.

\begin{thm}\label{thm:qcpbound}
    Let $B \in \mathbb{R}^{[3, n]}$ and $A\in\mathbb{R}^{[2, n]}$ be generalized row strictly diagonally dominant tensors with all positive diagonal entries, $q \in \mathbb{R}^n \backslash \mathbb{R}_{+}^n$ be any given vector. If $x \in \mathbb{R}^n$ is a solution of \eqref{eq:qcp}, then we have that $\min_{k\in \Omega(q)}\frac{-s(A)_k+\sqrt{s(A)_k^2-4s(B)_k q_k}}{2s(B)_k}\leq \|x\|_{\infty} \leq \max_{k\in \Omega(q)}\frac{-\delta(A)_k+\sqrt{\delta(A)_k^2-4\delta(B)_k q_k}}{2\delta(B)_k}$.
\end{thm}

\begin{proof}
The proof is similar to \cite{xu2019estimations}.
    According to Lemma \ref{lem:qcp1}, there exists $k \in \Omega(q)$ such that $x_k=\|x\|_{\infty}>0$ and $\left(Bx^2+Ax\right)_k+q_k=0$. Thus, we have 
    \begin{equation}
\begin{aligned}
0 & <\frac{-q_k}{\|x\|_{\infty}^{2}}
 =\left(B\left(\frac{x}{\|x\|_{\infty}}\right)^{2}+\frac{A\left(\frac{x}{\|x\|_{\infty}}\right)}{\|x\|_{\infty}}\right)_k \\
& =\sum_{i_2, i_3}^n B_{k i_2 i_3} \frac{x_{i_2} x_{i_3}}{\|x\|_{\infty}^{2}}+\frac{1}{\|x\|_{\infty}}\sum_{i_2} A_{ki_2} \frac{x_{i_2} x_{i_3}}{\|x\|_{\infty}}\\
& \leq B_{k k k}+r_k(B)_{+} + \frac{A_{k k}+r_k(A)_{+}}{\|x\|_{\infty}}=s(B)_k+\frac{s(A)_k}{\|x\|_{\infty}}.
\end{aligned}
\end{equation}

Similarly, we have that $\frac{-q_k}{\|x\|_{\infty}^{2}}\geq \delta(B)_k+\frac{\delta(A)_k}{\|x\|_{\infty}}$. Notice that $s(B)_k,s(A)_k,\delta(B)_k,\delta(A)_k$ are all positive. Finally, we have $\min_{k\in \Omega(q)}\frac{-s(A)_k+\sqrt{s(A)_k^2-4s(B)_k q_k}}{2s(B)_k}\leq \|x\|_{\infty} \leq \max_{k\in \Omega(q)}\frac{-\delta(A)_k+\sqrt{\delta(A)_k^2-4\delta(B)_k q_k}}{2\delta(B)_k}$.

\end{proof}

\section{General Higher-order Lotka-Volterra model}

A purely competitive higher-order Lotka-Volterra model is proposed in \cite{letten2019mechanistic}.
%In \cite{letten2019mechanistic}, they propose a Higher-order Lotka–Volterra competition model and take only competition among species into account. 
Here, we consider both, competition and cooperation, among species.
Inspired by \cite{letten2019mechanistic}, we study here a general higher-order Lotka-Volterra model:
\begin{equation}\label{eq:lvg}
\dot{x_{i}}=r_{i}x_{i}\left(1+\sum^{n}_{j = 1}a_{ij}x_{j}+\sum^{n}_{j,k = 1}b_{ijk}x_{j}x_{k}\right),
\end{equation}
 where $x_{i}$ denotes the density of species $i$; $r_{i}$ is the per-capita (per-species) intrinsic rate of increase of the
focal species, which is positive; $a_{ij}$ is the first-order coefficient, denoting $j$'s additive influence on $i$, and $b_{ijk}$ is the second-order (higher-order) coefficient, denoting $j$ and $k$'s joint non-additive influence on $i$, or alternatively $j$'s influence on $i$ correlated with a co-occurring third species $k$. We assume that $a_{ii}\leq 0$ and $b_{iii}\leq 0$, which describes competition within a species.

\begin{remark}
From the perspective of network science, if we ignore higher-order terms, then \eqref{eq:lvg} is a model on a graph. One can construct the corresponding digraph and label all the species as nodes in the digraph. Then, one links node $j$ to $i$ with the weight $a_{ij}$. If we further consider higher-order terms, the model is then based on a hypergraph. Simply speaking, a hypergraph is a higher-order network where one hyperedge can have multiple tails and heads. In our model \eqref{eq:lvg}, we have the last term for three-body interactions. For example, $b_{ijk}$ denotes $j$ and $k$'s joint influence on $i$. So one can create a hyperedge with the weight $b_{ijk}$, where $j$ and $k$ are the heads and $i$ is the tail. For a more detailed explanation of the concept of a hypergraph and the dynamical systems on it, interested readers may refer to \cite{bick2021higher}. For the definition of a directed hyperedge, one may refer to \cite{gallo1993directed}.
\end{remark}

\begin{lemma}
 The system \eqref{eq:lvg} is a positive system.
\end{lemma}

\begin{proof}
    It is sufficient to notice that if $x_i=0$, $\dot{x}_i=0$. Once the trajectory enters the boundary, since its derivative becomes zero, it will never leave the first orthant.
\end{proof}

System \eqref{eq:lvg} can be rewritten in tensor form as:
\begin{equation}\label{eq:lvgt}
\dot{x}=R\Dg(x)\left(\mathbf{1}+Ax+Bx^2\right),
\end{equation}
where $R=\Dg(r_1,\cdots,r_n)$, $A=(a_{ij})_{n\times n}$ and $B=(b_{ijk})_{n\times n \times n}$. It is worth mentioning that the diagonal entries $a_{ii}$ and $b_{iii}$ denote intra-specific interaction, while the off-diagonal $a_{ij}$ and $b_{ijk}$ denote inter-specific interaction.

It is easy to see that the origin is always an equilibrium of the system \eqref{eq:lvg}. Any other equilibrium point, if it exists, is given by the non-homogeneous polynomial equation system $\mathbf{1}+Ax+Bx^2=\mathbf{0}$.% in the tensor form or its subsystem which yields boundary equilibria.

From the definition of the polynomial complementarity problem, we have the following result.
\begin{cor}\label{cor:2}
    The set of all equilibria $x^*$ of system \eqref{eq:lvg} satisfying
    \begin{equation}\label{eq:boundary1}
        1+\sum_{j}a_{ij}x^*_j+\sum_{j,k}b_{ijk}x^*_j x^*_k\leq 0
    \end{equation}
    corresponds to the solution set $\operatorname{SOL}\left(F, -\mathbf{1}\right)$ of the polynomial complementarity problem \eqref{eq:pcp} with $F=-Ax-Bx^2$;
    while the set of all equilibria $x^*$ of system \eqref{eq:lvg} satisfying 
    \begin{equation}\label{eq:boundary2}
        1+\sum_{j}a_{ij}x^*_j+\sum_{j,k}b_{ijk}x^*_j x^*_k\geq 0
    \end{equation}
     corresponds to the solution set $\operatorname{SOL}\left(-F, \mathbf{1}\right)$ of the polynomial complementarity problem \eqref{eq:pcp}.
\end{cor}

\begin{proof}
The proof is straightforward and can be done by observation.
\end{proof}

\begin{remark}
Notice that the solution of the polynomial complementarity problem satisfies $1+\sum_{j}a_{ij}x^*_j+\sum_{j,k}b_{ijk}x^*_j x^*_k<0$ if $x_i^*=0$, and $1+\sum_{j}a_{ij}x^*_j+\sum_{j,k}b_{ijk}x^*_j x^*_k=0$ if $x_i^*>0$.
In this paper, we will show that \eqref{eq:boundary1} is related to the stability of some boundary equilibria (equilibria on the boundary of the first orthant). Thus, the polynomial complementarity problem will help us to find potential stable equilibria. Moreover, by further utilizing the result of Lemmas \ref{lem:pcp}-\ref{lem:pcp3}, we can know when the equilibrium set satisfying \eqref{eq:boundary1} is non-empty and bounded, or when it is a nonempty singleton. It is known that a classical Lotka-Volterra model is related to the linear complementarity problem \cite{takeuchi1996global,takeuchi1980existence}. Here, we show the relationship between PCP and the higher-order Lotka-Volterra model.
\end{remark} 

Generally, we have the following result regarding the origin.

\begin{thm}\label{thm:origin}
The origin is always an unstable equilibrium of \eqref{eq:lvg}.
\end{thm}

\begin{proof}
    It is straightforward to see that $\mathbf{0}$ is always a solution of $R\Dg(x)\left(\mathbf{1}+Ax+Bx^2\right)$.
For the equilibrium point $\mathbf{0}$, the corresponding Jacobian matrix is $\mathbf{J}_{(\mathbf{0})}=\Dg((r_1,\cdots,r_n)^\top)$.
Since $r_i>0$, the equilibrium point $\mathbf{0}_{n}$ is unstable.
\end{proof}

System \eqref{eq:lvg} can be written as 
\begin{equation}
    \dot{x}=\Dg(x)\left(R\mathbf{1}+RAx+RBx^2\right)= \Dg(x) F(x),
\end{equation}
where $F(x)=(F_1(x),\cdots,F_n(x))^\top$ and $F_i(x)=r_{i}+r_{i}\sum^{n}_{j = 1}a_{ij}x_{j}+r_{i}\sum^{n}_{j,k = 1}b_{ijk}x_{j}x_{k}$. We have the following result as a corollary of \cite[Theorem 5]{ye2021applications}. The definition and the concept are consistent with the paper \cite{ye2021applications}. For a set $\mathcal{M}$ with boundary, we denote the boundary as $\partial \mathcal{M}$, and the interior $\operatorname{Int}(\mathcal{M}) \triangleq \mathcal{M} \backslash \partial \mathcal{M}$. The concept of ``pointing inward'' is defined as in \cite[Definition 1]{ye2021applications}.

\begin{cor}\label{cor:poincare}
    Consider the system \eqref{eq:lvg}. Suppose that there exist constants $\hat{R}, \,\epsilon>0$  such that $\mathcal{W} \triangleq\left\{x:\|x\| \leq \hat{R},\, x_i \geq \epsilon \; \forall i=1, \ldots, n\right\}$ is a positively invariant set and $F(x)$ points inward at every $x \in \partial \mathcal{W}$. Then, there exists a feasible equilibrium in $\operatorname{Int}(\mathcal{W})$. Suppose further that for any equilibrium point $\bar{x} \in \mathcal{W}$:
\begin{equation}\label{eq:con1}
\begin{aligned}
\frac{\partial F_i(\bar{x})}{\partial x_i} & \leq l_{i i}<0 \\
\left|\frac{\partial F_i(\bar{x})}{\partial x_j}\right| & \leq l_{i j}, \forall i \neq j
\end{aligned}
\end{equation}
for some constant matrix $L=[l_{ij}]$, and all the leading principal minors of $-L$ are positive. Then, there is a unique feasible equilibrium $x^* \in \operatorname{Int}(\mathcal{W})$, and $x^*$ is locally exponentially stable. If the condition \eqref{eq:con1} holds for all $x \in \mathcal{W}$, then $x^*$ is globally 
exponentially stable in $\operatorname{Int}(\mathcal{W})$. 
\end{cor}

To utilize Corollary \ref{cor:poincare}, it is required that every orbit of  \eqref{eq:lvg} has an upper bound. However, such an upper bound may not always exist.

As a corollary of \cite[Theorem 4]{goh1978sector}, we have the following:
\begin{cor}\label{cor:sector}
Consider the system \eqref{eq:lvg}. Suppose that there exist constants $\hat{R}, \epsilon>$ 0 such that $\mathcal{W} \triangleq\left\{x:\|x\| \leq \hat{R}, x_i \geq \epsilon \quad \forall i=1, \ldots, n\right\}$ is a positive invariant set.
    Suppose there exists a constant matrix $G=[g_{ij}]$ such that in $\mathbb{R}^n_+\cap \mathcal{W}$
\begin{equation}\label{eq:jacbound}
    \begin{gathered}
\frac{\partial F_i(x)}{\partial x_i} \leqslant g_{i i}<0, \quad i=1,2, \ldots, m, \\
\left|\frac{\partial F_i(x)}{\partial x_j}\right| \leqslant g_{i j}, \quad i \neq j .
\end{gathered}
\end{equation}
If all the leading principal minors of $-G$ are positive, then the positive equilibrium $x^*$, if it exists, is globally asymptotically stable.
\end{cor}

Corollary \ref{cor:sector} is similar to Corollary \ref{cor:poincare}, but Corollary \ref{cor:sector} doesn't require that $F(x)$ points inward at every $x \in \partial \mathcal{W}$ and Corollary \ref{cor:sector} says nothing about the existence.

More precisely, we are able to provide the following results, which is a further corollary:
\begin{cor}\label{cor:gsglv}
Consider the system \eqref{eq:lvg} with $B$ being a non-positive tensor and $a_{ii}<0$ for all $i$. Suppose that there exist constants $\hat{R},\, \epsilon>0$ such that $\mathcal{W} \triangleq\left\{x:\|x\| \leq \hat{R}, x_i \geq \epsilon \quad \forall i=1, \ldots, n\right\}$ is a positive invariant set.
If there exists some coefficients $d_i$, $i=1,\cdots, n$ such that
$-d_i a_{ii}>\sum_{j\neq i}d_j|a_{ij}|+\hat{R}\sum_{j\neq i, 1\leq k\leq n}(d_j|b_{ijk}|+d_j|b_{ikj}|)$ for all $i$, then the positive equilibrium $x^*$, if it exists, is unique and globally asymptotically stable.
\end{cor}

\begin{proof}
We just need to show that the conditions in Corollary \ref{cor:sector} are satisfied. More precisely, one needs to show the Jacobian is bounded by a matrix $G$ and all the leading principal minors of $-G$ are positive.
    We have that $\frac{\partial F_i(x)}{\partial x_i}=r_i a_{ii}+r_i \sum_k b_{iki}x_k+r_i \sum_k b_{iik}x_k\leq r_i a_{ii}=G_{ii}<0$. On the other hand, $\left|\frac{\partial F_i(x)}{\partial x_j}\right|\leq r_i \{\sum_{j\neq i}|a_{ij}|+\hat{R}\sum_{j\neq i, 1\leq k\leq n}(|b_{ijk}|+|b_{ikj}|)\}=G_{ij}.$
    The condition $-d_i a_{ii}>\sum_{j\neq i}d_j|a_{ij}|+\hat{R}\sum_{j\neq i, 1\leq k\leq n}(d_j|b_{ijk}|+d_j|b_{ikj}|)$ guarantees that all the leading principal minors of $-G$ are positive, according to the Lemma 6 of \cite{ye2021applications}. Then, we can utilize corollary \ref{cor:sector}.
\end{proof}

\begin{remark}
For the condition $-d_i a_{ii}>\sum_{j\neq i}d_j|a_{ij}|+\hat{R}\sum_{j\neq i, 1\leq k\leq n}(d_j|b_{ijk}|+d_j|b_{ikj}|)$, it means that the diagonal entries of pairwise interaction dominates the whole dynamics.
If all $d_i=1, \forall i$, the matrix $-A$ is diagonally dominant with all positive diagonal elements, which is an $\mathcal{S}-$tensor. If $B$ is nonpositive, then $-B$ is also an $\mathcal{S}-$tensor.
    Thus, Corollary \ref{cor:gsglv} is compatible with Theorem \ref{thm:st}. 
\end{remark}

A slightly more general result is shown by the following.

\begin{cor}\label{cor:gsglv2}
Consider the system \eqref{eq:lvg} with $b_{iii}<0$ and $a_{ii}<0$ for all $i$. Suppose that there exist constants $\hat{R},\, \epsilon>0$ such that $\mathcal{W} \triangleq\left\{x:\|x\| \leq \hat{R}, x_i \geq \epsilon \quad \forall i=1, \ldots, n\right\}$ is a positive invariant set.
If there exists some coefficients $d_i$, $i=1,\cdots, n$ such that
$-d_i (a_{ii}+b_{iii}\epsilon)>\sum_{j\neq i}d_j|a_{ij}|+\hat{R}\sum_{(j,k)\neq (i,i)}(d_j|b_{ijk}|+d_j|b_{ikj}|)$ for all $i$, then the positive equilibrium $x^*$, if it exists, is unique and globally asymptotically stable.
\end{cor}

\begin{proof}
We just need to show that the conditions in Corollary \ref{cor:sector} are satisfied. More precisely, one needs to show the Jacobian is bounded by a matrix $G$ and all the leading principal minors of $-G$ are positive.
    We have that $\frac{\partial F_i(x)}{\partial x_i}=r_i a_{ii}+r_i \sum_k b_{iki}x_k+r_i \sum_k b_{iik}x_k\leq r_i (a_{ii}+b_{iii}\epsilon+ \sum_{b_{iki}>0} b_{iki} \hat{R} + \sum_{b_{iki}>0} b_{iik} \hat{R} )=G_{ii}<0$. On the other hand, $\left|\frac{\partial F_i(x)}{\partial x_j}\right|\leq r_i \{\sum_{j\neq i}|a_{ij}|+\hat{R}\sum_{j\neq i, 1\leq k\leq n}(|b_{ijk}|+|b_{ikj}|)\}=G_{ij}.$
    The condition $-d_i (a_{ii}+b_{iii}\epsilon)>\sum_{j\neq i}d_j|a_{ij}|+\hat{R}\sum_{(j,k)\neq (i,i)}(d_j|b_{ijk}|+d_j|b_{ikj}|)$ guarantees that all the leading principal minors of $-G$ are positive, according to the Lemma 6 of \cite{ye2021applications}. Then, we can utilize corollary \ref{cor:sector}.
\end{proof}

We can firstly use Theorem \ref{thm:st} to see whether there exists a positive equilibrium and then use Corollary \ref{cor:gsglv} and \ref{cor:gsglv2} to determine the uniqueness and global stability. That is to first use the following result.

% \begin{thm}\label{thm:gsglvt}
% Consider the system \eqref{eq:lvg} with $A,B$ both strictly diagonally dominant tensor with negative diagonal entries. The system has a unique positive equilibrium and it is globally asymptotically stable.
% \end{thm}
\begin{thm}\label{thm:gsglvt}
System \eqref{eq:lvg} has a unique equilibrium point if $-A$ and $-B$ are all $\mathcal{S}-$tensors.
\end{thm}

\begin{proof}
    The equation for a positive equilibrium is $-Ax-Bx^2=\mathbf{1}$.  From Theorem \ref{thm:st}, there is a unique positive equilibrium.
\end{proof}

\begin{remark}
    The result is obviously stronger than Corollary \ref{cor:poincare}. Here, we do not need a condition of a bounded positively invariant set. Recall that a strictly diagonally dominant tensor with all positive diagonal entries is $\mathcal{S}-$tensors. One observation from Theorem \ref{thm:gsglvt} and Corollary \ref{cor:gsglv2} is that if the negative diagonal entries of $a_{ii}, b_{iii}$ are sufficiently large (in absolute value) compared with the off-diagonals, then the system has a unique globally asymptotically stable equilibrium. Corollary \ref{cor:gsglv2} is conservative and diagonal entries of $a_{ii}, b_{iii}$ maybe sometimes huge. In the latter result of the Theorem \ref{thm:gsm} and \ref{thm:gsnpt}, we study a much more relaxed condition under some specific scenarios.
\end{remark}

In the following, we further use the perturbation method to provide results regarding the existence of positive equilibrium and its local stability. Let us now consider $0<\epsilon \ll 1$ and the perturbed system:
\begin{equation}\label{eq:perturbed1}
\dot{x}=R\Dg(x)\left(\mathbf{1}+Ax+\epsilon Bx^2\right).
\end{equation}
The system when $\epsilon=0$ is called an unperturbed system.
The unperturbed system \eqref{eq:perturbed1} with $\epsilon=0$ corresponds to the conventional Lotka-Volterra model on a graph, which is well-introduced in, e.g., \cite{takeuchi1996global,sb2010}, and from which we know when \eqref{eq:perturbed1} with $\epsilon=0$ has a hyperbolic equilibrium and whether it is stable.

\begin{thm}\label{thm:perturbation}
Consider the perturbed system \eqref{eq:perturbed1}. If the unperturbed system \eqref{eq:perturbed1} with $\epsilon=0$ has a hyperbolic equilibrium point $x^*$, then the perturbed system \eqref{eq:perturbed1} also has a hyperbolic equilibrium point $\tilde{x}^*$ in the vicinity of $x^*$. % of the unperturbed system. 
 Furthermore, if $x^*$ is locally stable, then $\tilde{x}^*$ is locally stable. Otherwise, if $x^*$ is unstable, then $\tilde{x}^*$ is unstable.
\end{thm}

\begin{proof}
The unperturbed system  can be represented as $\dot{x}=g(x)$ and the perturbed system  as $\dot{x}=G(x,\epsilon)$, with $G(x,0)=g(x)$. Let $x^*$ be a hyperbolic equilibrium of the unperturbed system. By definition of the equilibrium point, $G(x^*,0)=0$ and $\frac{\partial G}{\partial x}(x^*,0)=\frac{\partial g}{\partial x}(x^*)$. Due to the hyperbolicity of $x^*$,  $\frac{\partial G}{\partial z}(x^*,0)=\frac{\partial g}{\partial x}(x^*)$ has a nonvanishing determinant. By the implicit function theorem, there is a unique equilibrium in the neighborhood of $x^*$ for sufficiently small $\epsilon$. This equilibrium is also hyperbolic because of the continuous dependence of the eigenvalues of $\frac{\partial G}{\partial x}$ on $\epsilon$. Thus, the local stability of the equilibrium persists.
\end{proof}

In this section, we have provided abundant results regarding the existence of a positive equilibrium and its stability in some cases. Notice that there may also be some boundary equilibria (besides the origin). However, these non-zero boundary equilibria can be considered as a positive equilibrium of a subsystem, the one restricted to the boundary. Moreover, we can also use the solution of the polynomial complementarity problem to find some potential stable boundary equilibria.
Thus, all the results and techniques in this section are also applicable to that case.

\section{Cooperative Lotka-Volterra Model}

%\hjk{This section appears already elsewhere?}

We now consider a specific case of a higher-order cooperative Lotka-Volterra model for $n$ species:%, for $i\in \{1, 2, \cdots, n\}$ with $n$ the total number of the species:
\begin{equation}\label{model-1}
\begin{split}
\dot{x_{i}}&=r_{i}x_{i}\left(1-a_{ii}x_{i}+\sum_{j \neq i}a_{ij}x_{j}-b_{iii}x_{i}^{2}+\sum_{j,k \in Q_i}b_{ijk}x_{j}x_{k}\right)\\ &=r_{i}x_{i}L_{i}(x),    
\end{split}
  \end{equation}
where for the given focal species $i$, $Q_i=\{j,k\,|\,%~j,~k \text{ are not equal to $i$ at the same time}
(j,k)\neq(i,i)\}$.

Since in this section we consider cooperative systems, which represent the symbiosis among several species, it is natural to assume that the inter-specific interaction is non-negative. We further assume that the intra-specific interaction is non-positive, which is interpreted as competition within the species. Thus, all parameters in \eqref{model-1} are non-negative. We recall that a graph $G(A)$ is strongly connected if its adjacency matrix $A$ is irreducible. Throughout the section, we assume $A=[a_{ij}]$ is irreducible such that $G(A)$ is strongly connected.

System \eqref{model-1} can be also written in a tensor form \eqref{eq:lvgt}, where $A=[a_{ij}]$ is an irreducible Metzler matrix and $B=[b_{ijk}]$ is a Metzler tensor.

 Now we give some properties of the dynamical system \eqref{model-1}.

\begin{thm}\label{thm:1}
The system given by \eqref{model-1} is an irreducible monotone system in $\mathbb{R}_+^{n}$. Furthermore, if the system has an open and bounded positively invariant set $\mathcal{T}=\{x|\mathbf{0}_n< x < \mathbf{E}\}$, where $\mathbf{E}\in\mathbb R^n$, and if the model has a finite number of equilibria in the closure of $\mathcal{T}$, then the set of initial conditions in $\mathcal{T}$, such that the model does not converge to an equilibrium, is a set of Lebesgue measure zero. Furthermore, the model has a finite number of equilibria for a generic choice of parameters ($a_{ij},b_{ijk}$).
\end{thm}

\begin{proof}
Firstly, we calculate the Jacobian, which has components
\begin{equation}\label{Jac-1}
\frac{\partial \dot{x}_{i}}{\partial x_{i}}=r_{i}L_{i}+r_{i}x_{i}(-a_{ii}-2b_{iii}x_{i}+\sum_{k \neq i}b_{iik}x_{k}+\sum_{k \neq i}b_{iki}x_{k}),
\end{equation}
\begin{equation}\label{jac-2}
\frac{\partial \dot{x}_{i}}{\partial x_{j}}=r_{i}x_{i}(a_{ij}+\sum_{k \neq i}b_{ijk}x_{k}+\sum_{k\neq i}b_{ikj}x_{k}).
\end{equation}
We observe that the Jacobian is always an irreducible Metzler matrix (Definition 10.1 \cite{FB-LNS}). This ensures that \eqref{model-1} is an irreducible monotone system. 
Under the condition that the equilibrium set is finite and the system domain is bounded, by Lemma 2.3 of \cite{ye2022convergence}, the proof of the second statement follows.

To derive the equilibrium set and to see whether it is finite, one only needs to check whether the equation set of $L_i=0$ (which is a set of quadratic equations with multiple variables) for $i\leq n$ has a finite number of solutions. If we set all $a_{ij}=0$ for $i\neq j$ and $b_{ijk}=0$ for $j,k\in Q_i$, it is straightforward to check that the equilibrium set is finite. Also note that this parameter setting can only be used to check whether there is a particular choice such that the equation has a finite number of solutions. Since this algebraic question is different from the analysis of a system, it doesn't break the assumption that $A$ is irreducible. According to \cite[Theorem B.1 and Corollary B.2]{ye2022convergence}, since there exists a particular choice of parameters such that the equation has a finite number of solutions, if the parameters are generic and do not lie on a certain algebraic set of measure zero, then the equilibrium set is finite.
\end{proof}

\begin{remark}
Since the Jacobian of the system is an irreducible Metzler matrix, the system is indeed a cooperative system (see Chapter 4 and Definition 16 \cite{sb2010}). Theorem \ref{thm:1} requires that the positively invariant set of the system is open and bounded. One can easily check that the system is lower-bounded. Thus, this condition only requires that all solutions of the system have a supremum $\mathbf{E}$. It is worthwhile to mention that the system is not always upper-bounded and solutions may diverge to infinity due to the cooperation terms. 
 Since divergence to infinity is not natural in reality, we focus on the case when the system is upper-bounded throughout this paper. 
\end{remark}

Before we introduce further results, We recall that an irreducible Metzler matrix has the following property. 
\begin{lemma}[Theorem 10.14 in \cite{FB-LNS}]\label{lem:1}
If $M$ is an irreducible Metzler matrix, then $M$ is a Hurwitz matrix if and only if there is a vector $x> \mathbf{0}$ that satisfies $Mx <\mathbf{0}$.
\end{lemma}
Throughout this paper, we say that a species is ``a winner'' when it takes some positive value in the corresponding equilibrium or is ``a loser'' when it takes the zero value. We use the set $S$ to denote the set of agents of the winner faction. A boundary equilibrium $X^*$ is an equilibrium where $x_i\neq0$ for some $i\in S$ with non-empty $S$ and $x_j=0$ for the rest.

\begin{thm}\label{thm:mt1}
Consider system \eqref{model-1}, then the following hold:
\begin{itemize}
    \item [a)] the origin is always an equilibrium and is unstable;
    \item [b)] a boundary equilibrium, if it exists, is always unstable;
    \item [c)] if an all-species-coexistence equilibrium point $X^{*}=(x^{*}_{1},x^{*}_{2},\cdots ,x^{*}_{n})\in \mathcal T$ exists, and if the second-order cooperative coefficients $b_{ijk}$ ($j, k \in Q_i$) are sufficiently small such that $(\mathbf{J}_{(X^{*})}X^{*})<\mathbf{0}$ with $\mathbf{J}_{(X^{*})}$ the Jacobian of the system \eqref{model-1} at $X^{*}$, then %the all-species-coexistence equilibrium point
    $X^*$ is locally stable.
\end{itemize}

\end{thm}

\begin{proof}
For statement a), see Theorem \ref{thm:origin}.

Next, we prove statement b). By assumption, let \eqref{model-1} have a boundary equilibrium of $m<n$ winners. Consequently, the densities of the rest $n-m$ species are zero. Without loss of generality, we write the boundary equilibrium as $(x^{*}_{1},x^{*}_{2},\cdots ,x^{*}_{m}, 0, 0, \cdots, 0)=(X^{*}_m, \mathbf{0}_{n-m})$. Note that any other boundary equilibrium can be written in the previous form by index permutation. The corresponding Jacobian matrix of a boundary equilibrium is
$\mathbf{J}_{(X^{*}_m, \mathbf{0}_{n-m})}=\left(
\begin{array}{cc}
M & \Omega \\
\mathbf{0}_{(n-m)\times m} & D \\
\end{array}
\right),$
where $M$ is an irreducible Metzler matrix, $D$ is a diagonal matrix, with its diagonal entry $D_{i}=r_{i}(1+\sum_{j \in S}a_{ij}x_{j}+\sum_{j, k \in S}b_{ijk}x_{j}x_{k})>0$, where $S$ includes all the species of the winner faction. Hence, all boundary equilibrium points are unstable.

Finally, we investigate statement c). For the equilibrium point $X^{*}=(x^{*}_{1},x^{*}_{2},\cdots ,x^{*}_{n})$, the corresponding Jacobian matrix $\mathbf{J}_{(X^{*})}$ is an irreducible Metzler matrix. 
Moreover, we see that:
{\small\begin{equation*}
\begin{split}
    &(\mathbf{J}_{(X^{*})}X^{*})_{i}=r_{i}x^{*}_{i}L_{i}(x^*)\\
    &+r_{i}x^{*}_{i}\left(-a_{ii}x^{*}_{i}-2b_{iii}x_{i}^{*2}+\sum_{k \neq i}b_{iik}x^{*}_{i}x^*_{k}+\sum_{k \neq i}b_{iki}x^{*}_{i}x^*_{k}\right)\\
&+\sum_{j \neq i}r_{i}x^{*}_{i}\left(a_{ij}x^*_{j}+\sum_{k \neq i}b_{ijk}x^*_j x^*_{k}+\sum_{k\neq j}b_{ikj}x^*_{k} x^*_j\right).
\end{split}
\end{equation*}}

Recalling \eqref{model-1} and plugging $r_{i}x^*_{i}L_{i}(x^*)=0$  in the equation above, one can get $(\mathbf{J}_{(X^{*})}X^{*})_{i}=-r_{i}x^{*}_{i}+r_{i}x^{*}_{i}(-b_{iii}x_{i}^{*2}+\sum_{j,k \in Q_i}b_{ijk}x^*_{j}x^*_{k})$.
It follows that once there is equilibrium point $X^{*}=(x^{*}_{1},x^{*}_{2},\cdots ,x^{*}_{n})\in\mathcal T$, and if $b_{ijk}$ ($i, j \in Q_i$) is sufficiently small, \hma{then $(\mathbf{J}_{(X^{*})}X^{*})_{i}<0$. Thus, the Jacobian is Hurwitz and furthermore, the equilibrium point $X^{*}=(x^{*}_{1},x^{*}_{2},\cdots ,x^{*}_{n})$ is stable by Lemma \ref{lem:1}.}
\end{proof}

\begin{remark}\label{rem:1}
Now, we know from Corollary \ref{cor:2} that the equilibrium corresponds to the solution of the polynomial complementarity problem.
In light of Theorem \ref{thm:qcpbound}, we can adjust the parameters in the following way to make sure that the condition of statement c) in Theorem \ref{thm:mt1} holds surely. We decrease $b_{ijk}$ but also adjust $b_{iii}$ such that $\delta(-B)_k= -b_{k k k}-r_k(-B)_{-}$ remain unchanged. In this way, the upper bound of the equilibrium will be also unchanged. By continuously decreasing $b_{ijk}$, there must be a stable positive equilibrium. This technique is further applicable to the case of Theorem \ref{thm:allco} below.
\end{remark}

As a Corollary of Theorem \ref{thm:st}, we have the following result regarding the existence of a positive equilibrium for \eqref{model-1}.

\begin{thm} \label{thm:gsm}
    System \eqref{eq:lvgt} under the cooperative setting (i.e., \eqref{model-1}) has one unique positive equilibrium if $-B$ is an irreducible $\mathcal{M}$-tensor, $-A$ is an irreducible $\mathcal{M}$-matrix, and furthermore they are associated with the same positive vector $v$, i.e., $-Av>\mathbf{0}$ and $-Bv^2>\mathbf{0}$. From almost all initial conditions in $\{x|\mathbf{0}<x< x^*\}$, the solution converges to the unique positive equilibrium $x^*$. If $-Ax^*>\mathbf{0}, -Bx^*>\mathbf{0}$ further holds, then the $x^*$ is globally asymptotically stable.
\end{thm}

\begin{proof}
    Recall that an irreducible $\mathcal{M}$-tensor is a $\mathcal{S}$-tensor. Thus, the system has a unique positive equilibrium according to Theorem \ref{thm:gsglvt}.

Now, we prove the statement, from almost all initial conditions in $\{x|\mathbf{0}<x< x^*\}$, the solution converges to the unique positive equilibrium $x^*$. The cooperative system \eqref{model-1} is an irreducible monotone system. From the definition of irreducible monotone systems (section 2.3 \cite{ye2022convergence}), we have $\phi_t(x)< \phi_t(x^*)$ if $x< x^*$. Furthermore, if $x_i=0$, then $\Dot{x}_i>0$. Combining both cases, we confirm that $\{x|\mathbf{0}<x< x^*\}$ is a positively invariant set of the system \eqref{model-1}.

Then, we see that in $\{x|\mathbf{0}<x< x^*\}$, there are only one equilibrium, the $x^*$ on the boundary. According to Lemma 2.3 of \cite{ye2022convergence}, the solution converges to the unique positive equilibrium $x^*$ from almost all initial conditions in $\{x|\mathbf{0}<x< x^*\}$.

Next, we show the global stability when $-Ax^*>\mathbf{0}, -Bx^*>\mathbf{0}$. Define the set $U_1=\{x|\max \frac{x_i}{x_i^*}\geq 1\}$ and $U_2=\{x|x\leq x^*\}$.

 Then, we let $V_m= \max_i (\frac{x_i}{x^*_i})^2$ and $V= \max_i (\frac{x_i-x_i^*}{x^*_i})^2$. Further let $m=\arg\min_i (\frac{x_i}{x^*_i})^2$ and $t=\arg\min_i (\frac{x_i-x^*_i}{x^*_i})^2$. It holds $m=t$ because of $\frac{x_i-x^*_i}{x^*_i}=\frac{x_i}{x^*_i}-1$.

    Furthermore, it must hold $x_i=\frac{x_ix_i^*}{x_i^*}\leq V_m^{\frac{1}{2}}x_i^*.$ It holds as a equality when $i=m$, otherwise it holds as an strict inequality.

    Then, we have in the set $U_1$
    {\small\begin{equation*}
        \begin{split}
            \dot{V}&= \frac{2r_mx_m}{(x_m^*)^2}(x_m-x_m^*)(Ax+Bx^2+\mathbf{1})_m \\
            &= \frac{2r_mx_m}{(x_m^*)^2}(x_m-x_m^*)(\sum_j a_{mj} x_j+ \sum_{jk} b_{mjk} x_j x_k+1)\\
            &\leq \frac{2r_mx_m}{(x_m^*)^2}(x_m-x_m^*)(\sum_j a_{mj} x_j^*V_m^{\frac{1}{2}}+ \sum_{jk} b_{mjk} x^*_j x^*_k V_m+1)\\
            &\leq \frac{2r_mx_m}{(x_m^*)^2}(x_m-x_m^*)(\sum_j a_{mj} x_j^*V_m+ \sum_{jk} b_{mjk} x^*_j x^*_k V_m+1)\\
            &=\frac{2r_mx_m}{(x_m^*)^2}(x_m-x_m^*)(-V_m+1)\leq 0 \quad \forall x\in U_1.
        \end{split}
    \end{equation*}}

Otherwise, for the set $U_2$
    {\small\begin{equation*}
        \begin{split}
            \dot{V}
            &\leq \frac{2r_mx_m}{(x_m^*)^2}(x_m-x_m^*)(\sum_j a_{mj} x_j^*V_m^{\frac{1}{2}}+ \sum_{jk} b_{mjk} x^*_j x^*_k V^{\frac{1}{2}}_m+1)\\
            &=\frac{2r_mx_m}{(x_m^*)^2}(x_m-x_m^*)(-V^{\frac{1}{2}}_m+1)\leq 0 \quad \forall x\in U_2.
        \end{split}
    \end{equation*}}
Notice that $V$ is locally positive definite in $U_1$ and $U_2$, and $\dot{V}$ is negative semidefinite. The solution must converge to $\{x|\dot{V}=0\}$. In order to have $\dot{V}=0$, for every $i$, it must hold $i=m$ and $x_m=x_m^*$ which indicates that $\{x|\dot{V}=0\}$ only contains $x^*$. Thus, $x^*$ is asymptotically stable with a domain of attraction $U_1, U_2$.

Next, we further let $L_s=\min_i(\frac{x_i}{x^*_i})^2$ and $s=\arg \min_i(\frac{x_i}{x^*_i})^2$.

Furthermore, it must hold $x_i=\frac{x_ix_i^*}{x_i^*}\geq L_s^{\frac{1}{2}}x_i^*.$ It holds as a equality when $i=s$, otherwise it holds as an strict inequality.

Then, consider the set $U_3=\{x| \min_i \frac{x_i}{x^*_i}<1\}$, we have
    \begin{equation*}
        \begin{split}
            \dot{V}&= \frac{2r_sx_s}{(x_s^*)^2}(x_s)(Ax+Bx^2+\mathbf{1})_s \\
            &= \frac{2r_s x_s^2}{(x_s^*)^2}(\sum_j a_{sj} x_j+ \sum_{jk} b_{sjk} x_j x_k+1)\\
            &\geq \frac{2r_s x_s^2}{(x_s^*)^2}(\sum_j a_{sj} x_j^*L_s^{\frac{1}{2}}+ \sum_{jk} b_{sjk} x^*_j x^*_k L_s+1)\\
            &\geq \frac{2r_s x_s^2}{(x_s^*)^2}(\sum_j a_{sj} x_j^*L_s+ \sum_{jk} b_{sjk} x^*_j x^*_k L_s+1)\\
            &=\frac{2r_s x_s^2}{(x_s^*)^2}(-L_s+1)> 0 \quad \forall x\in U_3.
        \end{split}
    \end{equation*}
This guarantees that all trajectories in $U_3$ finally enter $U_1$. Thus, the $x^*$ is globally asymptotically stable.

\end{proof}

In \cite[Theorem 15.1.1]{hofbauer1998evolutionary} and \cite[ Theorem 14]{sb2010}, a sufficient condition of global stability of the unique positive equilibrium for a classical cooperative Lotka-Volterra model is provided. With the introduction of HOIs, positive equilibria of higher-order models are generally not unique. According to the definition of an irreducible monotone system, if there exists a positive equilibrium $X^*$, then $\{x|\mathbf{0}<x<X^*\}$ is positively invariant. This is a particular case when the system has an upper bound. Then, from Theorem \ref{thm:1}, we know that for almost all initial conditions within $\{x|\mathbf{0}<x<X^*\}$, the solution converges to a positive equilibrium (not necessarily $X^*$, could be another, i.e. multi-stability may occur).

\section{Higher-order Two-faction-competition Lotka-Volterra Model}

%\hjk{Also this appeared elsewhere?}

Now that we have described the dynamics of a single-faction model, we are ready to look into the two-faction model.

A system of competition between two factions is also known as a system with limited competition \cite{smith1986competing}.
%The system consists of two competing factions. 
The system generally follows a simple regulation, ``friends' friends are friends," and ``friends' enemies are enemies". Species within one faction cooperate with each other, while they compete with the species in the other faction. 

Consider the case where two factions of species (or agents), \hma{denoted by $x\in\mathbb R^m$ and $y\in\mathbb R^n$ respectively,} compete with each other but the agents inside the camp cooperate with each other. The corresponding model reads as:
%\hjk{It is a bit weird that agents of the same faction cooperate, except those of the same species, which are, by definition, in the same faction. Is this like saying, my friends help me, but I sabotage myself? I understand that depending on the situation this could indeed happen, bit it seems more natural to let $a_{ii}\neq0$}
{\small\begin{equation}\label{eq:1}
\begin{split}
\dot{x}_{i}&=r_{i}x_{i}\bigg(1-a_{ii}x_{i}+\sum_{j \neq i,j \in \mathbb{I}_{m}}a_{ij}x_{j}-\sum_{j \in \mathbb{I}_{n}}b_{ij}y_{j}\\
&-c_{iii}x_{i}^{2}+\sum_{j,k \in Q_i, j,k\in \mathbb{I}_{m}}c_{ijk}x_{j}x_{k}-\sum_{j,k \in \mathbb{I}_{n}}d_{ijk}y_{j}y_{k}\bigg)\\
&=r_{i}x_{i}L_{i}(x,y), ~~~~~i \in \mathbb{I}_{m},    
\end{split}
\end{equation}
\begin{equation}\label{eq:2}
\begin{split}
\dot{y}_{i}&=\hat{r}_{i}y_{i}\bigg(1-\hat{a}_{ii}y_{i}+\sum_{j \neq i, j\in\mathbb{I}_{n}}\hat{a}_{ij}y_{j}-\sum_{j \in \mathbb{I}_{m}}\hat{b}_{ij}x_{j}  \\
&-\hat{c}_{iii}y_{i}^{2}+\sum_{j,k \in Q_i,j,k\in \mathbb{I}_{n}}\hat{c}_{ijk}y_{j}y_{k}-\sum_{j,k \in \mathbb{I}_{m}}\hat{d}_{ijk}x_{j}x_{k} \bigg) \\
&=\hat{r}_{i}y_{i}\hat{L}_{i}(x,y), ~~~~~i \in \mathbb{I}_{n},    
\end{split}
\end{equation}}\par
\noindent where $\mathbb{I}_{m}=\{1, 2, \cdots, m\}$, $\mathbb{I}_{n}=\{1, 2, \cdots, n\}$, and $m,n$ are the total number of species in each faction respectively; $a_{ij},b_{ij},\hat{a}_{ij},\hat{b}_{ij}$ are the first-order coefficients and $c_{ijk},d_{ijk},\hat{c}_{ijk},\hat{d}_{ijk}$ are the higher-order coefficients. As for the modeling setup, we assume that all the parameters are non-negative so that all the intra-faction interaction is non-negative except the negative self-competition of one agent (species) with itself, while the inter-faction interaction is non-positive. We further assume that there is no multi-body interaction, where head agents are from different factions, i.e., there are no crossed terms $x_iy_j$. Besides relevant scenarios in ecology, this two-faction competition model also has the potential to describe political campaigns between two parties or business competitions between two corporations \cite{wijeratne2009bifurcation,hidayati2021stability}. We further assume throughout the section that the matrix 
\begin{equation}\label{eq:matrixA}
A=\left(
\begin{array}{cc}
(a_{ij})_{m\times m} & (b_{ij})_{m\times n} \\
(\hat{a}_{ij})_{n\times m} & (\hat{b}_{ij})_{n\times n} \\
\end{array}
\right)
\end{equation}
is irreducible.
In general, the model described by \eqref{model-1} can be regarded as a special case of \eqref{eq:1}-\eqref{eq:2}, where one faction is empty. For simplicity, throughout this section, we call $x$ the first faction and $y$ the second. Define the notation $z=(x,y)^\top$ and let $z_0$ denote the initial condition.

The system \eqref{eq:1}-\eqref{eq:2} can be rewritten in a tensor form as:
\begin{equation}\label{eq:lvgt2}
\dot{z}=R\Dg(z)\left(\mathbf{1}+Az+Bz^2\right),
\end{equation}
where $A$ is in the form of \eqref{eq:matrixA}, and $B_{:,:,k}=\left(
\begin{array}{cc}
(c_{ijk})_{m\times m} & (d_{ijk})_{m\times n} \\
(\hat{c}_{ijk})_{n\times m} & (\hat{d}_{ijk})_{n\times n} \\
\end{array}
\right),$ where the index $i,j$ can vary and $k$ is fixed. This notation ($:$) has the same meaning as in Matlab$^\text{\textregistered}$, i.e., denoting all indices in the dimension. For example, $A(:,:,k)$ is the $k$-th page of the three-dimensional array $A$.

\begin{thm}\label{thm:irr}
The system \eqref{eq:1}-\eqref{eq:2} is an irreducible monotone system in $\mathbb{R}_+^{n+m}$.
\end{thm}

\begin{proof}
%Firstly, we calculate the partial derivatives.
%\begin{equation}\label{eq:j1}
%\frac{\partial \dot{x}_{i}}{\partial x_{i}}=r_{i}L_{i}+r_{i}x_{i}(-a_{ii}-2c_{iii}x_{i}+\sum_{k \neq i}c_{iik}x_{k}+\sum_{k \neq i}c_{iki}x_{k}),
%\end{equation}
%\begin{equation}\label{eq:j2}
%\frac{\partial \dot{x}_{i}}{\partial x_{j}}=r_{i}x_{i}(a_{ij}+\sum_{k \neq i}c_{ijk}x_{k}+\sum_{k \neq i,k\neq j}c_{ikj}x_{k}),
%\end{equation}
%\begin{equation}\label{eq:j3}
%\frac{\partial \dot{x}_{i}}{\partial y_{j}}=r_{i}x_{i}(-b_{ij}-\sum_{k \in \mathbb{I}_{n}}d_{ijk}y_{k}-\sum_{k\neq j}d_{ikj}y_{k}).
%\end{equation}
%We further obtain that the $\frac{\partial \dot{y}_{i}}{\partial y_{i}},\frac{\partial \dot{y}_{i}}{\partial y_{j}},\frac{\partial \dot{y}_{i}}{\partial x_{j}}$ are symmetric to those above. 
Firstly, one can check that the Jacobian of \eqref{eq:1}-\eqref{eq:2} is of the form $\mathbf{J}_z=\left(
\begin{array}{cc}
M_{1} & T_{1} \\
T_{2} & M_{2} \\
\end{array}
\right),$
where $M_{1}$, $M_{2}$ are irreducible Metzler matrices, and $T_{1}$, $T_{2}$ are non-positive matrices (it is irreducible as long as $z> \mathbf{0}_{n+m}$). It follows that $\mathbf{J}$ can be permuted into an irreducible Metzler matrix $\mathbf{\tilde{J}}= \left(
\begin{array}{cc}
M_{1} & -T_{1} \\
-T_{2} & M_{2} \\
\end{array}
\right)$ via the matrix $P=\left(
\begin{array}{cc}
I_{m} & \mathbf{0}_{m\times n} \\
\mathbf{0}_{n\times n} & -I_{n} \\
\end{array}
\right)$, so the system \eqref{eq:1}-\eqref{eq:2} is irreducible monotone.
\end{proof}

\begin{remark}
Theorem \ref{thm:irr} tells us, in particular, that
%One can see that 
the Jacobian of the two-faction competition model can be permuted into an irreducible Metzler matrix. However, if one would deal with competition among more than two factions, the permutation may be no longer possible. For example, if we consider 3 factions, the structure of the Jacobian with 3 factions is analog to the case of \cite[Theorem 1]{gracy2022endemic}. A permutation is not possible for the same reason in \cite{gracy2022endemic}. In addition, the two-faction system may still diverge to infinity because of the cooperation terms. 
\end{remark}

In the following Theorems \ref{thm:4}-\ref{thm:bec2}, we list all the possible equilibria of the model and study their stability. 

\begin{thm}\label{thm:4}
Consider the system \eqref{eq:1}-\eqref{eq:2}, the origin is always an equilibrium and is unstable.
\end{thm}

\begin{proof}
See Theorem \ref{thm:origin}.
\end{proof}
%\hjk{I would suggest writting the above proof as: straightforward.}

\begin{thm} \label{thm:1fac}
Consider the system \eqref{eq:1}-\eqref{eq:2}, and assume that a one-faction-wins-all boundary equilibrium $(X^*,\mathbf{0}_n)=(x^{*}_{1},x^{*}_{2},\cdots ,x^{*}_{m}, 0, 0, \cdots, 0)$ or $(\mathbf{0}_m,Y^*)=(0, 0, \cdots, 0, y^{*}_{1},y^{*}_{2},\cdots ,y^{*}_{n})$ exists. Either equilibrium is locally stable whenever the coefficients of the first-order ($b_{ij}, \hat{b}_{ij}\,|\,i\neq j$) and second-order  ($d_{ijk},\hat{d}_{ijk}\,|\, j,k\in Q_i$) competitive terms from the loser faction are sufficiently large such that $D_{i}=\hat{r}_{i}(1-\sum_{j \in S}\hat{b}_{ij}x_{j}^{*}-\sum_{j, k \in S}\hat{d}_{ijk}x_{j}^{*}x_{k}^{*})<0,\qquad i\in \mathbb{I}_{n},$ and $X^*$ or $Y^*$
%$(x^{*}_{1},x^{*}_{2},\cdots ,x^{*}_{m})$ or $(y^{*}_{1},y^{*}_{2},\cdots ,y^{*}_{n})$ 
is a stable all-species-coexistence equilibrium point of the sub-cooperative-system from the winner faction when ignoring the loser faction.
\end{thm}
%\hjk{Suggestion for the above theorem: Consider \eqref{eq:1}-\eqref{eq:2}, and assume that a one-faction-wins-all boundary equilibrium $(X^*,0_n)=(x^{*}_{1},x^{*}_{2},\cdots ,x^{*}_{m}, 0, 0, \cdots, 0)$ or $(0_m,Y^*)=(0, 0, \cdots, 0, y^{*}_{1},y^{*}_{2},\cdots ,y^{*}_{n})$ exists. Either equilibrium is locally stable whenever the coefficients of the first-order ($b_{ij}, \hat{b}_{ij}\,|\,i\neq j$) and second-order  ($d_{ijk}\,|\,\hat{d}_{ijk},j,k\in Q_i$) competitive terms from the loser faction are sufficiently large and $(x^{*}_{1},x^{*}_{2},\cdots ,x^{*}_{m})$ or $(y^{*}_{1},y^{*}_{2},\cdots ,y^{*}_{n})$ is a stable all-species-coexistence equilibrium point of the sub-cooperative-system from the winner faction when ignoring the loser faction.}

\begin{proof}
Without loss of generality, we first investigate the case when the first faction is the winner. The equilibrium is then $(x^{*}_{1},x^{*}_{2},\cdots ,x^{*}_{m}, 0, 0, \cdots, 0)=(X^{*}, \mathbf{0}_{n})$. By plugging the equilibrium into the Jacobian, we obtain that %the corresponding Jacobian matrix is
$\mathbf{J}_{(X^{*}, \mathbf{0}_{n})}=\left(
\begin{array}{cc}
M & \Omega \\
\mathbf{0}_{n\times m} & D \\
\end{array}
\right),$
where $M$ is an irreducible Metzler matrix and represents the Jacobian of the sub-cooperative-system from the winner faction on an all-species-coexistence equilibrium point, $D$ is a diagonal matrix, and its diagonal entry reads $D_{i}=\hat{r}_{i}(1-\sum_{j \in S}\hat{b}_{ij}x_{j}^{*}-\sum_{j, k \in S}\hat{d}_{ijk}x_{j}^{*}x_{k}^{*}),\qquad i\in \mathbb{I}_{n}.$ %where $S$ denotes the set of agents from the winner faction and in this case  $\mathbb{I}_{m}$. 
Since the Jacobian is an upper-triangular block matrix, we know that the Jacobian is Hurwitz when all $D_i<0, i\in \mathbb{I}_{n}$ and the matrix $M$ is Hurwitz, which further implies that the coefficients of the first- and second-order competitive terms from the loser faction are sufficiently large and  $(x^{*}_{1},x^{*}_{2},\cdots ,x^{*}_{m})$ is a stable all-species-coexistence equilibrium point of the sub-cooperative-system from the winner faction when ignoring the loser faction. We recall that the second condition is satisfied when the cooperative HOIs terms are sufficiently small for the winners.
The proof, for the case when the second faction is the winner, is exactly the same and thus omitted here.
\end{proof}

We then consider the following Lemma.
\begin{lemma}[Corollary 3.2 and Proposition 3.5 \cite{smith1986competing}]\label{lem:smith}
    Consider the system \eqref{eq:1}-\eqref{eq:2}, if $(x^{*}_{1},x^{*}_{2},\cdots ,x^{*}_{m}, 0,0,\cdots ,0)$ and $(0,0,\cdots ,0,y^{*}_{1},y^{*}_{2},\cdots ,y^{*}_{n})$ both exist and are both unstable, then there exists a positive all-species-coexistence equilibrium $(\tilde{x}^{*}_{1},\tilde{x}^{*}_{2},\cdots ,\tilde{x}^{*}_{m}, \tilde{y}^{*}_{1},\tilde{y}^{*}_{2},\cdots ,\tilde{y}^{*}_{n})$ with $\tilde{x}^{*}_{i}\leq x_i^*, \tilde{y}^{*}_{i}\leq y_i^*$ for arbitrary $i$, and $[\mathbf{0}_{(n+m)\times (n+m)},(x^{*}_{1},x^{*}_{2},\cdots ,x^{*}_{m},y^{*}_{1},y^{*}_{2},\cdots ,y^{*}_{n})]$ is a bounded positively invariant set.
\end{lemma}

Alternatively, we provide another result regarding the existence of a positive all-species-coexistence equilibrium. This result doesn't require the pre-knowledge about the boundary equilibrium.

\begin{cor}
    Consider the system \eqref{eq:1}-\eqref{eq:2}, if $-A$ is an $\mathcal{H}^+$ matrix, $-B$ is an $\mathcal{H}^+$ tensor, and they are associated with the same positive vector $v$, i.e., $-Av>\mathbf{0}$ and $-Bv^2>\mathbf{0}$, then there exist one unique positive equilibrium.
\end{cor}

\begin{proof}
    It suffices to notice that a $\mathcal{H}^+$ tensor is an $\mathcal{S}$-tensor.
\end{proof}

Notice that $-A$ and $-B$ have non-negative diagonal elements. The definition of an $\mathcal{H}^+$ tensor requires all positive diagonal elements. So, both settings are compatible with each other. 
    Furthermore, once the negative diagonal entries of $a_{ii}, b_{iii}$ are sufficiently large (in absolute value) compared with the off-diagonals, then the system has a unique globally asymptotically stable equilibrium from Theorem \ref{thm:gsglvt} and Corollary \ref{cor:gsglv2}.

The stability of a positive all-species-coexistence equilibrium can be checked by the following Theorem.

\begin{thm}\label{thm:allco}
Consider the system \eqref{eq:1}-\eqref{eq:2}, if the all-species-coexistence equilibrium $(x^{*}_{1},x^{*}_{2},\cdots ,x^{*}_{m}, y^{*}_{1},y^{*}_{2},\cdots ,y^{*}_{n})$ exists, then it is locally stable when the coefficients of all the first-order competitive terms ($b_{ij},\hat{b}_{ij},i\neq j$) and all the second-order terms except the self-influence term ($c_{ijk},d_{ijk},\hat{c}_{ijk},\hat{d}_{ijk},j,k\in Q_i$) are sufficiently small such that $\left(P\mathbf{J}_{(X^{*}, Y^{*})}P\left(
\begin{array}{c}
X^{*} \\
Y^{*} \\
\end{array}
\right)\right)_{i}<0$ holds.
\end{thm}

\begin{proof}
We know that the Jacobian $\mathbf{{J}}$ of the model \eqref{eq:1}-\eqref{eq:2} can be permutated into an irreducible Metzler matrix $\mathbf{\tilde{J}}$. Thus, $\mathbf{{J}}$ and $\mathbf{\tilde{J}}$ have the same eigenvalues. Therefore, $\mathbf{{J}}$ is Hurwitz if $\mathbf{\tilde{J}}$ is Hurwitz.
%Then, letting the notation 
Letting $Z^*=(X^*,Y^*)^\top$, we have 
{\small\begin{equation*}
\begin{aligned}
&\left(\mathbf{\tilde{J}}_{(X^{*}, Y^{*})}\left(
\begin{array}{c}
X^{*} \\
Y^{*} \\
\end{array}
\right)\right)_{i}=r_{i}x^{*}_{i}L_{i}(z^*)\\
&+r_{i}x^{*}_{i}\Big(-a_{ii}x^{*}_{i}-2c_{iii}x_{i}^{*2}+\sum_{k \neq i}c_{iik}x^{*}_{i}x^*_{k}+\sum_{k \neq i}c_{iki}x^{*}_{i}x^*_{k}\Big)\\
&~~~+\sum_{j \neq i}r_{i}x^{*}_{i}\Big(a_{ij}x^*_{k}+\sum_{k \neq i}c_{ijk}x^*_{j}x^*_{k}+\sum_{k\neq i}c_{ikj}x^*_{j}x^*_{k}\Big)\\
&~~~+\sum_{j \neq i}r_{i}x^{*}_{i}\Big(b_{ij}y^*_{j}+\sum_{k \in \mathbb{I}_{n}}d_{ijk}y^*_{j}y^*_{k}+\sum_{k \in \mathbb{I}_{n}}d_{ikj}y^*_{j}y^*_{k}\Big)\\
&=-r_{i}x^{*}_{i}+r_{i}x^{*}_{i}\Big(-c_{iii}x_{i}^{*2}+\sum_{j,k \in Q_i}c_{ijk}x^*_{j}x^*_{k}\\
&+3\sum_{j,k \in Q_i}d_{ijk}y^*_{j}y^*_{k}+2\sum_{j \neq i}b_{ij}y^*_{j}\Big), \quad i \in \mathbb{I}_{m}.
\end{aligned}
\end{equation*}}

On the other hand, for $i=j+m, j\in \mathbb{I}_{n}$, 
$\left(\mathbf{\tilde{J}}_{(X^{*}, Y^{*})}\left(
\begin{array}{c}
X^{*} \\
Y^{*} \\
\end{array}
\right)\right)_{i}
=-\hat{r}_{i}y^{*}_{i}
+\hat{r}_{i}y^{*}_{i}\Big(-\hat{c}_{iii}y_{i}^{*2}
+\sum_{j,k \in Q_i}\hat{c}_{ijk}y^*_{j}y^*_{k}
+3\sum_{j,k \in Q_i}\hat{d}_{ijk}x^*_{j}x^*_{k}+2\sum_{j \neq i}\hat{b}_{ij}x^*_{j}\Big).
$

According to Lemma \ref{lem:1}, $\mathbf{\tilde{J}}$ is Hurwitz if $b_{ij},c_{ijk},d_{ijk},\hat{b}_{ij},\hat{c}_{ijk},\hat{d}_{ijk}$ are sufficiently small such that $\left(\mathbf{\tilde{J}}_{(X^{*}, Y^{*})}\left(
\begin{array}{c}
X^{*} \\
Y^{*} \\
\end{array}
\right)\right)_{i}<0$ for all $i$. 
%This completes the proof.
\end{proof}

\begin{remark}
In \cite{goh1979stability} (Theorem 3 and (A8) in Appendix) and \cite{smith1986competing} (Theorem 3.8), a sufficient condition for the global stability of a positive equilibrium (all-species-coexistence equilibrium) for the abstract system $\dot{N}_i=N_i F_i\left(N_1, N_2, \ldots, N_m\right), \quad i=1,2, \ldots, m$ was provided by Lyapunov function \cite{goh1979stability} or monotone system theory \cite{smith1986competing}. Since both systems proposed in our paper can be written in such a form, the results in \cite{goh1979stability} are also valid for our models. However, Theorem 3.8 \cite{smith1986competing} can not be applied to our system because HOIs break condition 3.2 in \cite{smith1986competing}. The results in \cite{goh1979stability,smith1986competing} miss the possibility of bistability and what kind of role the HOIs play in the species' coexistence. Our paper fills this gap. Moreover, the arguments similar to remark \ref{rem:1} are applicable for the case of Theorem \ref{thm:allco}. For the global stability, the corollary \ref{cor:gsglv2} is directly applicable.
\end{remark}

\begin{thm}\label{thm:bec}
Consider the system \eqref{eq:1}-\eqref{eq:2}, the boundary equilibrium $( x_1^*,\ldots,x_l^*,\mathbf{0}_{m-l},\mathbf{0}_n),\, l< m$ or $(\mathbf{0}_{m}, y_1^*,\ldots,y_p^*,\mathbf{0}_{n-p}),\, p<n$, if it exists, is unstable.
\end{thm}
%\hjk{to save some space, I think we should always write $0_\bullet$ to denote a vector of zeroes.}\hma{However, the number of 0 is arbitrary here, it may cause some misunderstanding when using vector here}

\begin{proof}
We firstly investigate the first case $( x_1^*,\ldots,x_l^*,\mathbf{0}_{m-l},\mathbf{0}_n), l< m$.
The corresponding Jacobian matrix is
$\mathbf{J}_{(X^{*},  0, \cdots, 0)}=\left(
\begin{array}{cc}
M & \Omega \\
\mathbf{0}_{(m+n-l)\times l} & D \\
\end{array}
\right)$,
where $M$ is an irreducible Metzler matrix and represents the Jacobian of the sub-cooperative-system from the winner faction on a boundary equilibrium point, and $D$ is a diagonal matrix.
We know from Theorem \ref{thm:1}, that $M$ is unstable and thus the equilibrium $( x_1^*,\ldots,x_l^*,\mathbf{0}_{m-l},\mathbf{0}_n), l< m$ is unstable.
The proof, for the second case, is exactly analogous.
\end{proof}

\begin{remark}
Theorem \ref{thm:bec} shows that any equilibrium where the winners are a strict subset of a single faction, is unstable because any of those equilibria can be permuted into the equilibrium $( x_1^*,\ldots,x_l^*,\mathbf{0}_{m-l},\mathbf{0}_n), l< m$ or $(\mathbf{0}_{m}, y_1^*,\ldots,y_l^*,\mathbf{0}_{n-l}), p<n$ by index permutation.
\end{remark}

\begin{thm}\label{thm:bec2}%\hjk{Probably $M^*$ is not the best choice for this equilibrium.}
Consider the system \eqref{eq:1}-\eqref{eq:2}, the boundary equilibrium $M^{*}=(x_{1}^{*},\cdots,  x_{a}^{*}, \mathbf{0}_{m-a}, y_{1}^{*},\cdots,  y_{b}^{*},\mathbf{0}_{n-b})$, if it exists, is locally stable when the coefficients of the first- ($b_{ij},\hat{b}_{ij},i\neq j$) and second-order competitive terms ($d_{ijk},\hat{d}_{ijk},j,k\in Q_i$) of the losers are sufficiently large, such that $1+\sum_{j \in S}a_{ij}x^*_{j}-\sum_{j \in S}b_{ij}y^*_{j}
   +\sum_{j,k \in S}c_{ijk}x^*_{j}x^*_{k}-\sum_{j,k \in S}d_{ijk}y^*_{j}y^*_{k}<0$ and $1+\sum_{j \in S}\hat{a}_{ij}y^*_{j}-\sum_{j \in S}\hat{b}_{ij}x^*_{j}
   +\sum_{j,k \in S}\hat{c}_{ijk}y^*_{j}y^*_{k}-\sum_{j,k \in S}\hat{d}_{ijk}x^*_{j}x^*_{k}<0$, and $(x_{1}^{*},\cdots,  x_{a}^{*}, y_{1}^{*},\cdots,  y_{b}^{*})$ is a stable all-species-coexistence equilibrium point of the sub-system of the winners when ignoring the loser agents.
\end{thm}

\begin{proof}
Firstly, we perform index permutation, and $M^{*}$ can be permuted as $N^{*}=(x_{1}^{*},\cdots,  x_{a}^{*}, y_{1}^{*},\cdots,  y_{b}^{*},\mathbf{0}_{m-a},\mathbf{0}_{n-b})$.
The corresponding Jacobian matrix after the index permutation is
$\mathbf{J}_{N^*}=\left(
\begin{array}{cc}
M & \Omega \\
\mathbf{0}_{(m+n-a-b)\times(a+b)} & D \\
\end{array}
\right),$
where $M$ is an irreducible Metzler matrix and represents the Jacobian of the sub-system of all winners on an all-species-coexistence equilibrium point, $D$ is a diagonal matrix, with entries 
{\small\begin{equation*}
    \begin{split}
   D_{i}&=r_{i}L_{i}(z^*)=r_{i}(1+\sum_{j \in S}a_{ij}x^*_{j}-\sum_{j \in S}b_{ij}y^*_{j}\\
   &+\sum_{j,k \in S}c_{ijk}x^*_{j}x^*_{k}-\sum_{j,k \in S}d_{ijk}y^*_{j}y^*_{k}) ;
\end{split}
\end{equation*}}\par
\noindent if $i$ denotes the loser agent in the first faction. Otherwise, if $i$ denotes the loser agent in the second faction, $D_i$ takes a similar form. %$D_{i}=\hat{r}_{i}\hat{L}_{i}(z^*)=\hat{r}_{i}(1+\sum_{j \in S}\hat{a}_{ij}y^*_{j}-\sum_{j \in S}\hat{b}_{ij}x^*_{j}+\sum_{j,k \in S}\hat{c}_{ijk}y^*_{j}y^*_{k}-\sum_{j,k \in S}\hat{d}_{ijk}x^*_{j}x^*_{k})$.
In order to have all negative eigenvalues, the first- and second-order cooperative terms of the losers must be sufficiently small, so that all $D_i<0$, and $(x_{1}^{*},\cdots,  x_{a}^{*}, y_{1}^{*},\cdots,  y_{b}^{*})$ must be a stable all-species-coexistence equilibrium point of the sub-system of the winners when ignoring the loser agents so that $M$ is Hurwitz. %We recall that the second condition is satisfied when the HOIs except the self-competition term is sufficiently small for the winners.
\end{proof}

\begin{remark}
  Theorem \ref{thm:bec2} further implies that any equilibrium, where the winners are in the different camps, but not all of them, is stable under the same condition as Theorem \ref{thm:bec2} because any of those equilibria can be permuted into the equilibrium $M^{*}$ or $N^{*}$ by index permutation. In light of Corollary \ref{cor:2} and Lemma \ref{lem:qcp1}, by following a similar procedure of Remark \ref{rem:1}, we can confirm that there are always cases that conditions in Theorem \ref{thm:1fac}, \ref{thm:allco}, and \ref{thm:bec2} are satisfied.
\end{remark}

\section{Purely competitive higher-order Lotka-Volterra}
In this section, we consider the case of competition among $n$ species:
\begin{equation}\label{eq:lvgpc}
\dot{x_{i}}=r_{i}x_{i}\left(1-\sum^{n}_{j = 1}a_{ij}x_{j}-\sum^{n}_{j,k = 1}b_{ijk}x_{j}x_{k}\right),
\end{equation}
with $i=1,\ldots, n$, and where $a_{ij}$ and $b_{ijk}$ are all non-negative. In this scenario, all species have no friends and they all compete with each other.

\begin{lemma}
    The system \eqref{eq:lvgpc} is a positive system. Every positive orbit has an upper bound.
\end{lemma}

\begin{proof}
    It is sufficient to check that $\dot{x}_i=0$ when $x_i=0$ and $\dot{x}_i=r_{i}x_{i}\left(1-\sum^{n}_{j = 1}a_{ij}x_{j}-\sum^{n}_{j,k = 1}b_{ijk}x_{j}x_{k}\right)\leq r_i x_i- r_i a_{ii} x_i^2=r_ix_i(1-a_{ii}x_i)<0$ when $x_i>\frac{1}{a_{ii}}$.
\end{proof}

The existence of one unique positive equilibrium is guaranteed by Theorem \ref{thm:gsglvt}. We recall that 
a strictly diagonally dominant tensor with positive diagonal elements is a $\mathcal{S}$-tensor, see \cite{wang2019existence}. For example, we consider $A$ and $B$ are strictly diagonally non-positive dominant tensors and $-A$ and $-B$ are with positive diagonal elements. Then $-A$ and $-B$ are $\mathcal{S}$-tensors. The Theorem \ref{thm:gsglvt} is thus directly applicable. Once there is a positive equilibrium, the positive equilibrium is globally attractive when the condition of corollary \ref{cor:gsglv} is fulfilled. It is worthwhile to mention that in the case of pure competition, we don't need the assumption of the existence of an upper bound in corollary \ref{cor:gsglv}. The general upper bound in the case is $R=max_i\frac{1}{a_{ii}}$. This leads to the following result:

\begin{cor}
    Consider the system \eqref{eq:lvgpc}. Let $\hat{R}=\max_i\frac{1}{a_{ii}}$.
If there holds
$-d_i a_{ii}>\sum_{j\neq i}d_j|a_{ij}|+\hat{R}\sum_{j\neq i, 1\leq k\leq n}(d_j|b_{ijk}|+d_j|b_{ikj}|)$ for all $i$, then the positive equilibrium $x^*$, if it exists, is unique and globally asymptotically stable.
\end{cor}

Otherwise, the existence and stability can be checked with the following result, which is a stronger result.

\begin{thm}\label{thm:gsnpt}
    System \eqref{eq:lvgt} under the competitive setting (i.e., \eqref{eq:lvgpc}) has one unique positive equilibrium if $-B$ is an irreducible nonnegative $\mathcal{H}^+$-tensor, $-A$ is an irreducible nonnegative $\mathcal{H}^+$-matrix. The equilibrium point $x^*$ is globally asymptotically stable.
\end{thm}

\begin{proof}
    Notice that an irreducible nonnegative tensor is a $\mathcal{S}$-tensor and for any $v>\mathbf{0}$, $Av>\mathbf{0}$ and $Bv^2>\mathbf{0}$. Thus, the system has a unique positive equilibrium according to Theorem \ref{thm:gsglvt}. Define the set $U_1=\{x|x\geq x^*\}$ and $U_2=\{x|\min \frac{x_i}{x_i^*}<1\}$.

 Then, we let $V_m= \min_i (\frac{x_i}{x^*_i})^2$ and $V= \min_i (\frac{x_i-x_i^*}{x^*_i})^2$. Further let $m=\arg\min_i (\frac{x_i}{x^*_i})^2$ and $t=\arg\min_i (\frac{x_i-x^*_i}{x^*_i})^2$. It holds $m=t$ because of $\frac{x_i-x^*_i}{x^*_i}=\frac{x_i}{x^*_i}-1$.

    Furthermore, it must hold $-x_i=-\frac{x_ix_i^*}{x_i^*}\leq -V_m^{\frac{1}{2}}x_i^*.$ It holds as a equality when $i=m$, otherwise it holds as an strict inequality.

    Then, we have in the set $U_1$
    {\small\begin{equation*}
        \begin{split}
            \dot{V}&= \frac{2r_mx_m}{(x_m^*)^2}(x_m-x_m^*)(Ax+Bx^2+\mathbf{1})_m \\
            &= \frac{2r_mx_m}{(x_m^*)^2}(x_m-x_m^*)(\sum_j a_{mj} x_j+ \sum_{jk} b_{mjk} x_j x_k+1)\\
            &\leq \frac{2r_mx_m}{(x_m^*)^2}(x_m-x_m^*)(\sum_j a_{mj} x_j^*V_m^{\frac{1}{2}}+ \sum_{jk} b_{mjk} x^*_j x^*_k V_m+1)\\
            &\leq \frac{2r_mx_m}{(x_m^*)^2}(x_m-x_m^*)(\sum_j a_{mj} x_j^*V_m+ \sum_{jk} b_{mjk} x^*_j x^*_k V_m+1)\\
            &=\frac{2r_mx_m}{(x_m^*)^2}(x_m-x_m^*)(-V_m+1)\leq 0 \quad \forall x\in U_1.
        \end{split}
    \end{equation*}}

Otherwise, for the set $U_2$
    {\small\begin{equation*}
        \begin{split}
            \dot{V}
            &\leq \frac{2r_mx_m}{(x_m^*)^2}(x_m-x_m^*)(\sum_j a_{mj} x_j^*V_m^{\frac{1}{2}}+ \sum_{jk} b_{mjk} x^*_j x^*_k V^{\frac{1}{2}}_m+1)\\
            &=\frac{2r_mx_m}{(x_m^*)^2}(x_m-x_m^*)(-V^{\frac{1}{2}}_m+1)\leq 0 \quad \forall x\in U_2.
        \end{split}
    \end{equation*}}
Notice that $V$ is locally positive definite in $U_1$ and $U_2$, and $\dot{V}$ is negative semidefinite. The solution must converge to $\{x|\dot{V}=0\}$. In order to have $\dot{V}=0$, for every $i$, it must hold $i=m$ and $x_m=x_m^*$ which indicates that $\{x|\dot{V}=0\}$ only contains $x^*$. Thus, $x^*$ is asymptotically stable with a domain of attraction $U_1, U_2$.

Next, we further let $L_s=\max_i(\frac{x_i}{x^*_i})^2$ and $s=\arg \max_i(\frac{x_i}{x^*_i})^2$.

Furthermore, it must hold $-x_i=-\frac{x_ix_i^*}{x_i^*}\geq -L_s^{\frac{1}{2}}x_i^*.$ It holds as a equality when $i=s$, otherwise it holds as an strict inequality.

Then, consider the set $U_3=\{x| x<x^*\}$, we have
    \begin{equation*}
        \begin{split}
            \dot{V}&= \frac{2r_sx_s}{(x_s^*)^2}(x_s)(Ax+Bx^2+\mathbf{1})_s \\
            &= \frac{2r_s x_s^2}{(x_s^*)^2}(\sum_j a_{sj} x_j+ \sum_{jk} b_{sjk} x_j x_k+1)\\
            &\geq \frac{2r_s x_s^2}{(x_s^*)^2}(\sum_j a_{sj} x_j^*L_m^{\frac{1}{2}}+ \sum_{jk} b_{sjk} x^*_j x^*_k L_m+1)\\
            &\geq \frac{2r_s x_s^2}{(x_s^*)^2}(\sum_j a_{sj} x_j^*L_m+ \sum_{jk} b_{sjk} x^*_j x^*_k L_m+1)\\
            &=\frac{2r_s x_s^2}{(x_s^*)^2}(-L_m+1)> 0 \quad \forall x\in U_3.
        \end{split}
    \end{equation*}
This guarantees that all trajectories in $U_3$ finally enter $U_1\cup U_2$. Thus, the $x^*$ is globally asymptotically stable.
\end{proof}

While the positive equilibrium corresponds to the winner-share-all consequence, the following result is about the winner-take-all case. This result is an extension of  \cite[Corollary 7.4]{zeeman1995extinction} and \cite[Theorem 15]{sb2010} for the classical competitive Lotka-Volterra model.

\begin{thm}\label{thm:wta}
    Consider the system \eqref{eq:lvgpc} and the following conditions:
    \begin{itemize}
        \item[(A)] $\frac{1}{ a_{j j}}<\frac{1}{a_{i j}}, \quad 1 \leq i<j \leq n$, and
        \item[(B)] $\frac{1}{ a_{j j}}>\frac{1}{ a_{i j}}, n \geq i>j \geq 1$.
    \end{itemize}
    Then, $\left(\frac{-a_{11}+\sqrt{a_{11}^2+4b_{111}}}{ 2b_{111}}, 0, \ldots, 0\right)$ is globally attracting on $\mathbb{R}_{n}^+$.
\end{thm}

\begin{proof}
    We have that
    \begin{equation*}
\dot{x_{i}}\leq r_{i}x_{i}\left(1-\sum^{n}_{j = 1}a_{ij}x_{j}\right),
\end{equation*}

The higher-order system \eqref{eq:lvgpc} is upper-bounded by a classical competitive Lotka-Volterra counterpart (\cite[chapter 5]{sb2010}). We know that $\left(\frac{1}{a_{11}}, 0, \ldots, 0\right)$ is globally attracting on $\mathbb{R}_{n}^+$ for a classical competitive Lotka-Volterra under the conditions (A) and (B) (\cite[Theorem 15]{sb2010}). If $i\neq 1$, then $x_i$ will converge to zero since the classical model converges to zero. If $i=1$, then $\dot{x_{i}}$ is an autonomous system $\dot{x_{i}}=1-a_{ii}x_i-b_{iii}x_i^2$ plus a vanishing perturbation. The solution of $\dot{x_{i}}=1-a_{ii}x_i-b_{iii}x_i^2$ will converge to $\left(\frac{-a_{11}+\sqrt{a_{11}^2+4b_{111}}}{ 2b_{111}}, 0, \ldots, 0\right)$. We then see that $\frac{\frac{-a_{11}+\sqrt{a_{11}^2+4b_{111}}}{ 2b_{111}}} {\frac{1}{a_{11}}} =\frac{a_{11}}{a_{11}+\sqrt{a_{11}^2+4b_{111}}}\leq 1$. This stable equilibrium is also upper-bounded by the equilibrium of its classical competitive Lotka-Volterra counterpart.

Thus, the solution of the system \eqref{eq:lvgpc} will converge to $\left(\frac{-a_{11}+\sqrt{a_{11}^2+4b_{111}}}{ 2b_{111}}, 0, \ldots, 0\right)$.
\end{proof}

\begin{remark}
    Notice that Theorem \ref{thm:wta} describes the case when the first species is the winner. However, via a simple index permutation, we can also describe the case when the other species is the winner.
\end{remark}

There are also some winner-share-all cases, but not all species are winners.

\begin{thm}\label{thm:wsa2}
Consider the system \eqref{eq:lvgpc}. The boundary equilibrium $M^{*}=(x_{1}^{*},\cdots,  x_{a}^{*}, \mathbf{0}_{n-a})$, if it exists, is locally stable when the coefficients of the first- ($a_{ij}, i\neq j$) and second-order competitive terms ($d_{ijk}, j,k\in Q_i$) of the losers are sufficiently large such that $1-\sum_{j \in S}a_{ij}x^*_{j}
   -\sum_{j,k \in S}c_{ijk}x^*_{j}x^*_{k}<0$, and moreover $(x_{1}^{*},\cdots,  x_{a}^{*})$ is a stable positive equilibrium point of the sub-system of the winners when ignoring the loser agents.
\end{thm}

\begin{proof}
    The proof is similar to the proof of Theorem \ref{thm:bec2}.
\end{proof}

\begin{remark}
    Since there are multiple sub-systems of \eqref{eq:lvgt}, each subsystem may have a stable positive equilibrium. In general, multi-stability may occur.
\end{remark}

\section{{Numerical example}}
In this section, we use some numerical examples to illustrate our analytical results. For the simulation setup, we randomly pick all the parameters from the set $[0,10]$. Furthermore, the initial condition is randomly chosen from $[0,10]$. The cooperative Lotka-Volterra model can be considered as a sub-system of the two-faction-competition Lotka-Volterra model. Thus, we omit the simulations for that case.

Firstly, we conduct simulations for the case of competition between 2 factions. We assume a total of $5$ species, $2$ in one faction and $3$ in the other. 
%We set the first $2$ species in one faction and the rest $3$ species in the other. 
Figures \ref{fig:1faction}-\ref{fig:2zero} include all the possible results we can observe and are three typical examples for each case. Figure \ref{fig:1faction} is in line with Theorem \ref{thm:1fac}, while Figure \ref{fig:allcoexist} corresponds to Theorem \ref{thm:allco}. Similarly, Figure \ref{fig:2zero} reflects the analytical results of Theorem \ref{thm:bec2}. Since all other equilibria are unstable, we don't observe that the solution converges towards them in the simulations. From figure \ref{fig:1faction} and \ref{fig:bista}, we confirm the bistability properties of the two-faction system. %That is because stability conditions for the different equilibria can hold at the same time.
Furthermore, to highlight the influence of HOIs, the system parameters and initial conditions are the same in simulations \ref{fig:allcoexist} and \ref{fig:2zero} except that we change $d_{1jk}$ to $5d_{1jk}$, respectively. As we increase the $d_{1jk}$, according to the Theorem \ref{thm:bec2}, the first faction dies out, which is indeed in line with the simulation results. This shows that our theory can also be seen as a manipulation strategy to adjust the winner species. Since HOIs usually denote the indirect interaction in ecology, HOIs are potentially more suitable to adjust than pairwise direct interaction.
Finally, from simulations, we observe that large self-competition terms ($a_{ii},\hat{a}_{ii},c_{iii},\hat{c}_{iii}$) will improve the chance that the system doesn't diverge to infinity. This is because diagonally dominant tensors with all positive diagonal entries are $\mathcal{S}$-tensors, which potentially yields a globally stable positive equilibrium.

%\hjk{It would be nice if all figures were homogeneous}

%\hjk{I guess you do not have to call them "species" in the plots... $x_i$, $y_i$ would be enough}

\begin{figure*}
\begin{subfigure}[t]{0.24\linewidth}
    \centering
\includegraphics[height=3.5cm]{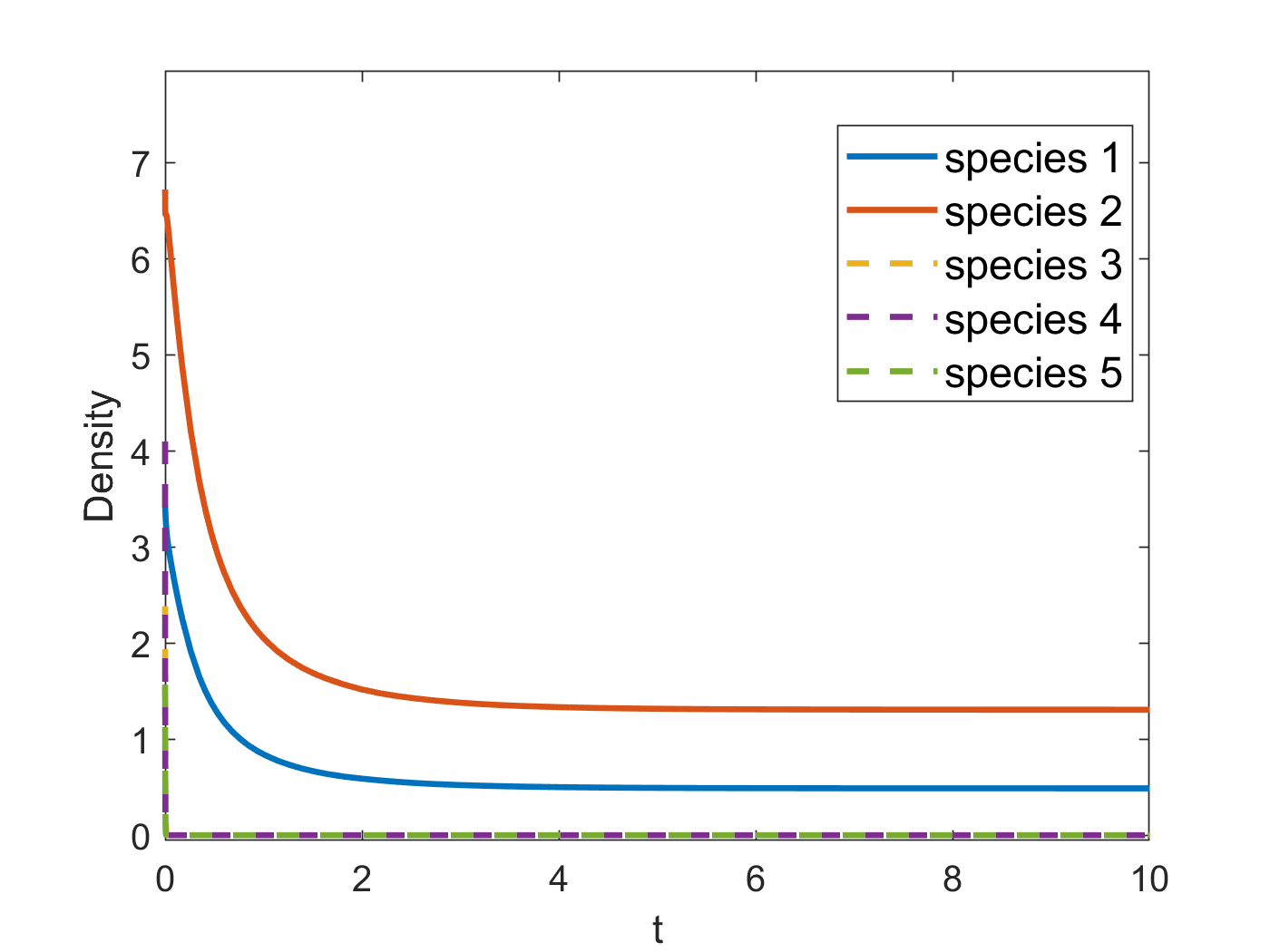}
\captionsetup{width=.95\textwidth}
    \caption{}
    \label{fig:1faction}
  \end{subfigure}
  \begin{subfigure}[t]{0.24\linewidth}
    \centering
\includegraphics[height=3.5cm]{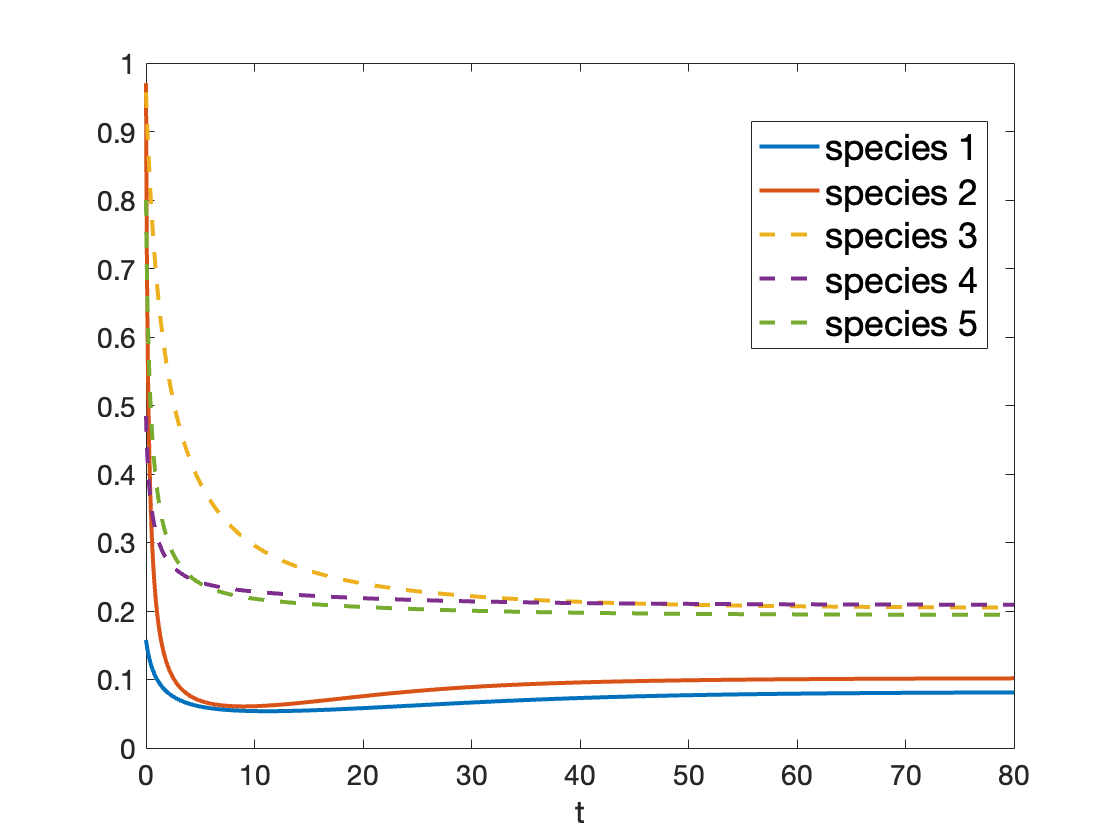}
\captionsetup{width=.95\textwidth}
    \caption{}
    \label{fig:allcoexist}
  \end{subfigure}
  \begin{subfigure}[t]{0.24\linewidth}
    \centering
\includegraphics[height=3.5cm]{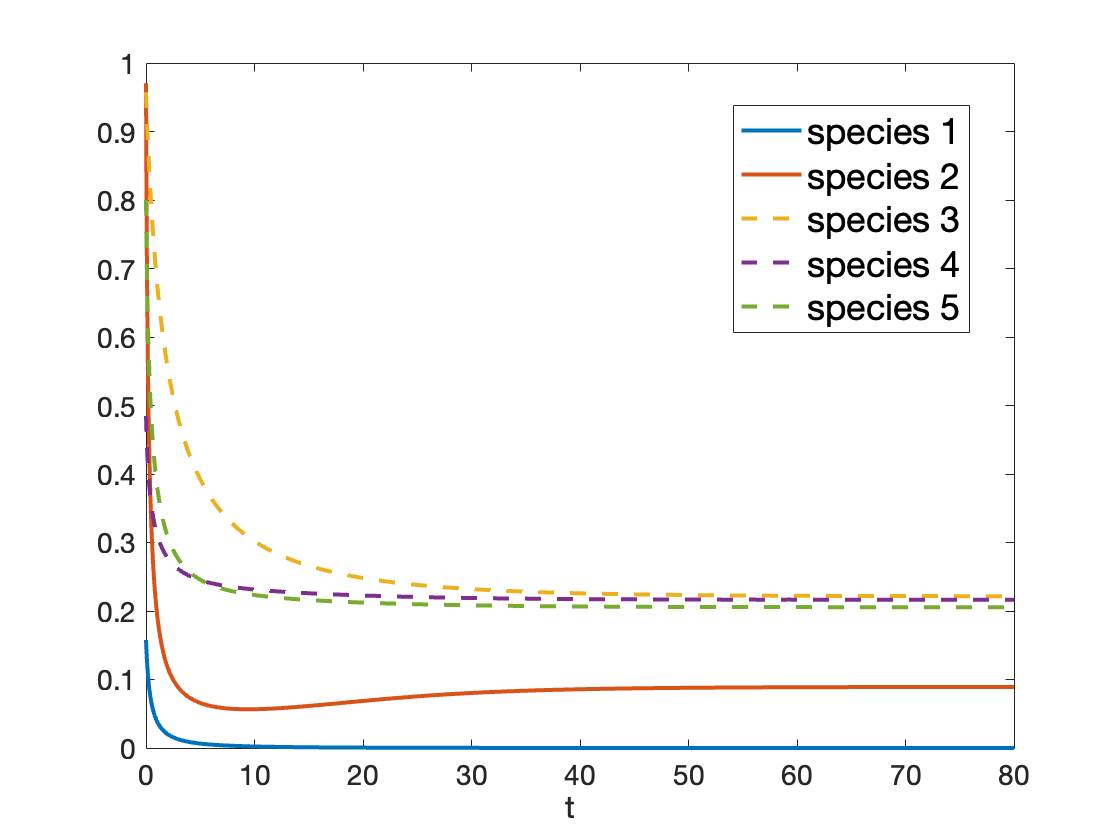}
\captionsetup{width=.95\textwidth}
    \caption{}
    \label{fig:2zero}
\end{subfigure}
\begin{subfigure}[t]{0.24\linewidth}
    \centering
\includegraphics[height=3.5cm]{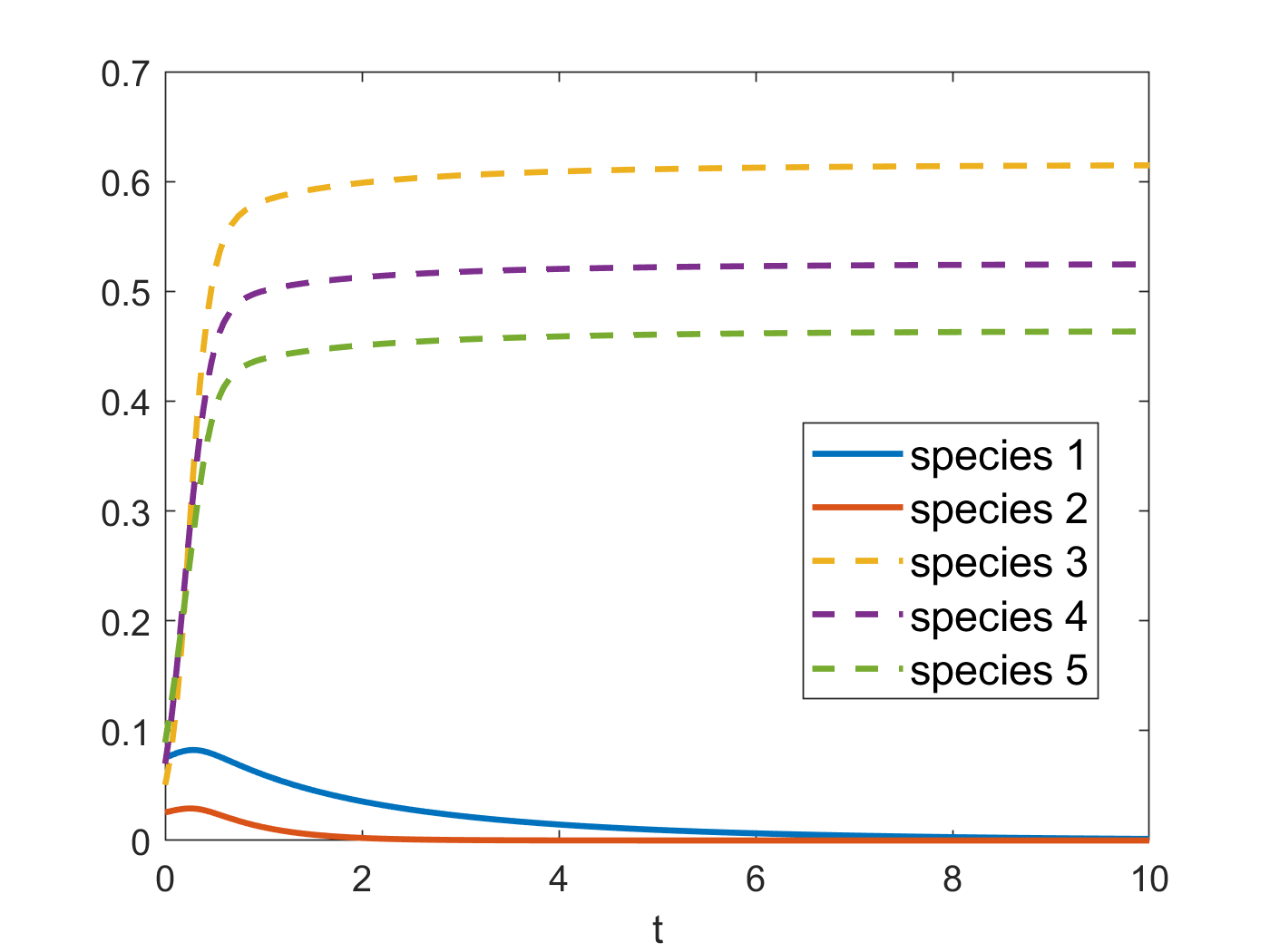}
\captionsetup{width=.95\textwidth}
    \caption{}
    \label{fig:bista}
\end{subfigure}
\caption{
(a) Only one faction wins and all members in the faction are the winners.
(b) All species coexist.
(c) Species from different factions win but some species die out. (d)  From a different initial condition but the same parameters with (a), the solution converges to the different one-faction-wins-all boundary equilibrium. That is, bistability is reflected in an interchange of the winner faction.%The winner factions interchange.} 
}
\end{figure*}

Secondly, we consider a purely competitive case for 5 species. All species compete with each other. Fig \ref{fig:wta} corresponds to Corollary \ref{cor:gsglv}. Fig \ref{fig:wsa1} is in line with Theorem \ref{thm:wsa2}. Moreover, Fig \ref{fig:wsa2} aligns with Theorem \ref{thm:wta}. Then, we further focus on the winner-share-all cases when not all species are winners. From Figure \ref{fig:wsabis} and \ref{fig:wsabis2}, we see that multi-stability may occur. With the same system parameters, from a low-level initial condition, species 2 and 5 are winners; while from a high-level initial condition, species 4 and 5 are winners.

\begin{figure*}
\begin{subfigure}[t]{0.33\linewidth}
    \centering
\includegraphics[height=5cm]{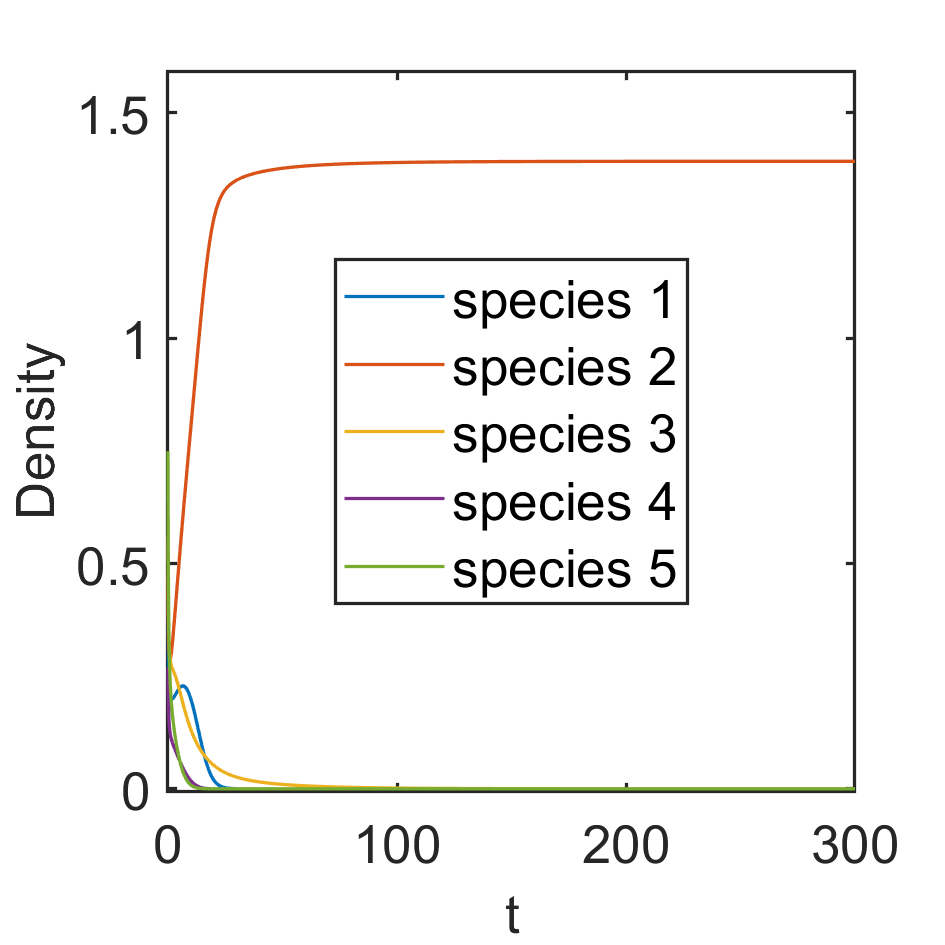}
\captionsetup{width=.95\textwidth}
    \caption{}
    \label{fig:wta}
  \end{subfigure}
  \begin{subfigure}[t]{0.33\linewidth}
    \centering
\includegraphics[height=5cm]{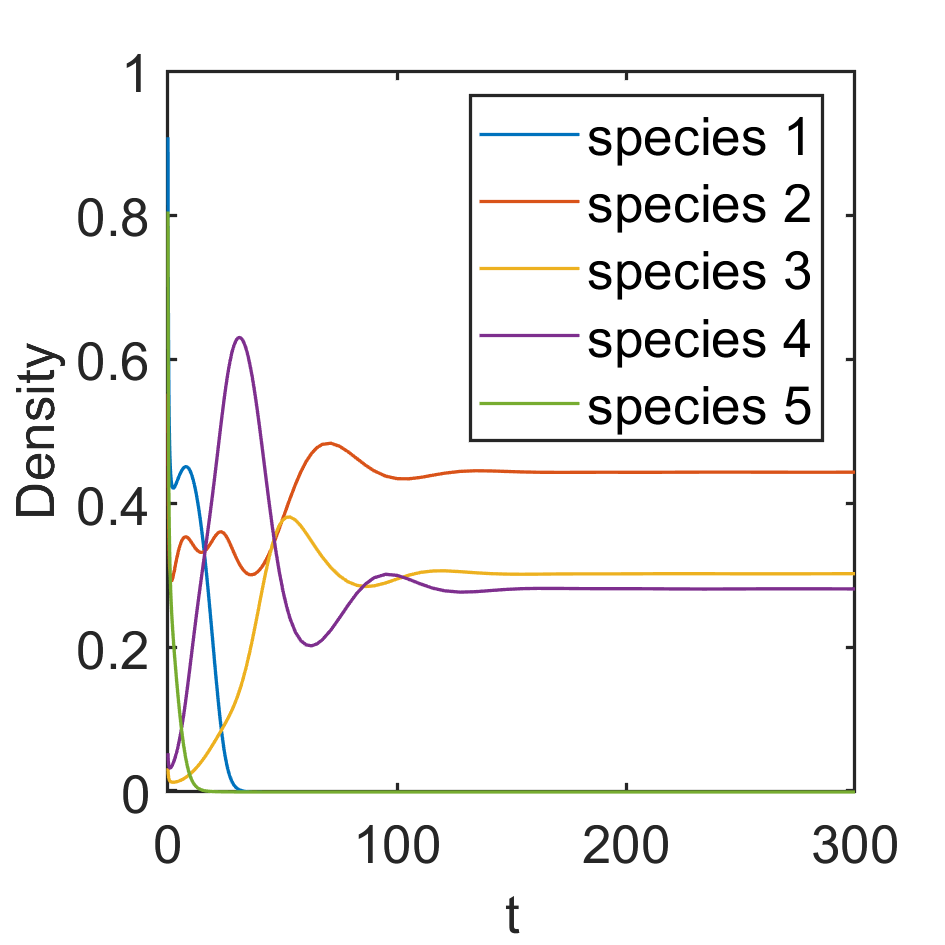}
\captionsetup{width=.95\textwidth}
    \caption{}
    \label{fig:wsa1}
  \end{subfigure}
  \begin{subfigure}[t]{0.33\linewidth}
    \centering
\includegraphics[height=5cm]{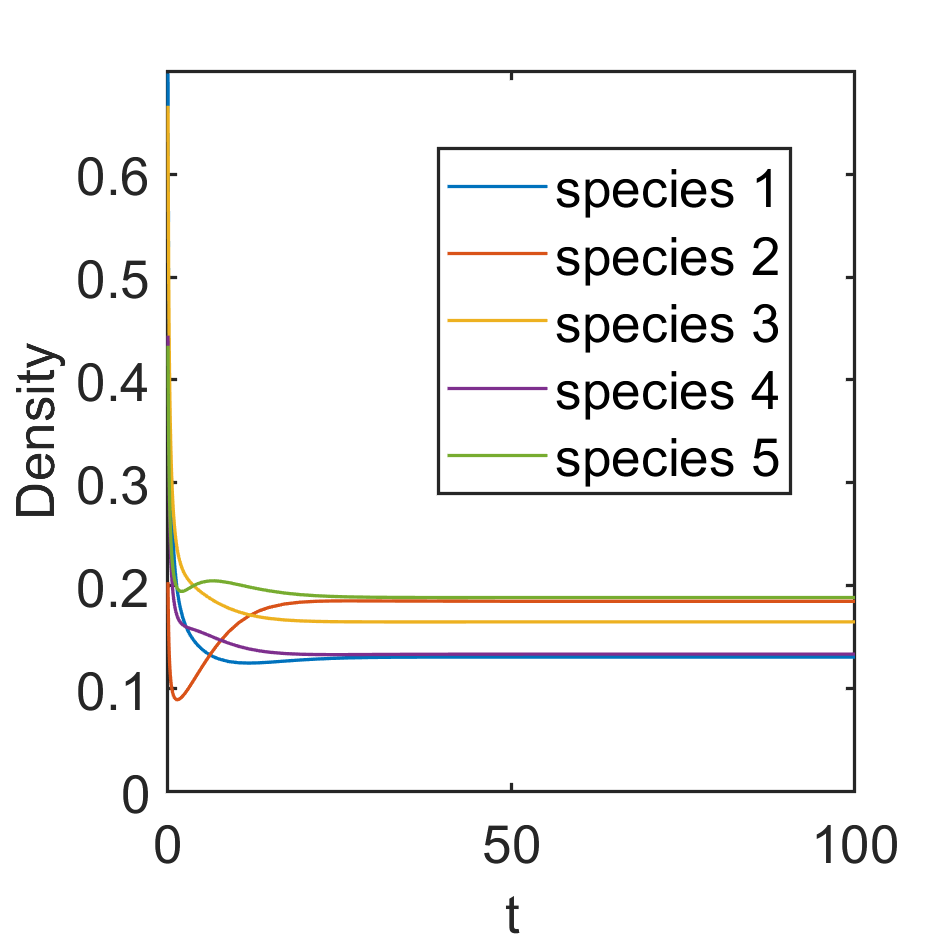}
\captionsetup{width=.95\textwidth}
    \caption{}
    \label{fig:wsa2}
\end{subfigure}
\caption{
(a) Only one species wins and the winner takes all.
(b) Winners share all, but not all species are winners.
(c) All species coexist.  
}
\end{figure*}

\begin{figure*}
\begin{subfigure}[t]{0.5\linewidth}
    \centering
\includegraphics[height=5cm]{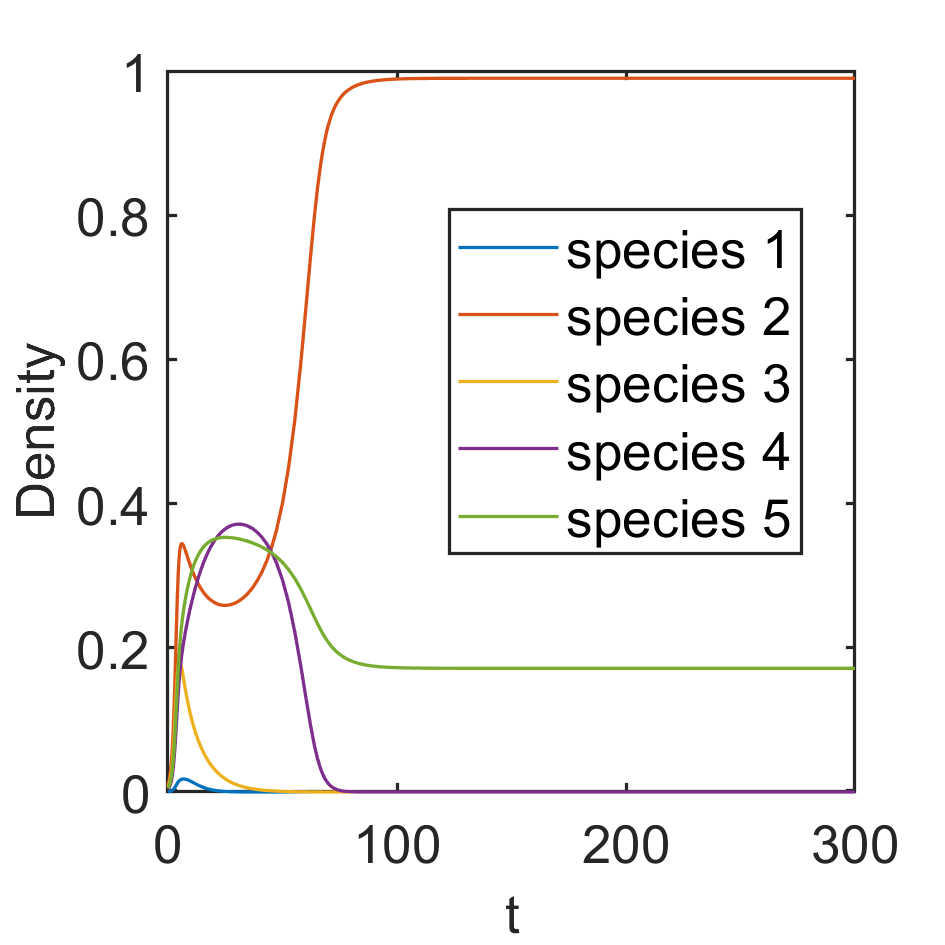}
\captionsetup{width=.95\textwidth}
    \caption{}
    \label{fig:wsabis}
  \end{subfigure}
  \begin{subfigure}[t]{0.5\linewidth}
    \centering
\includegraphics[height=5cm]{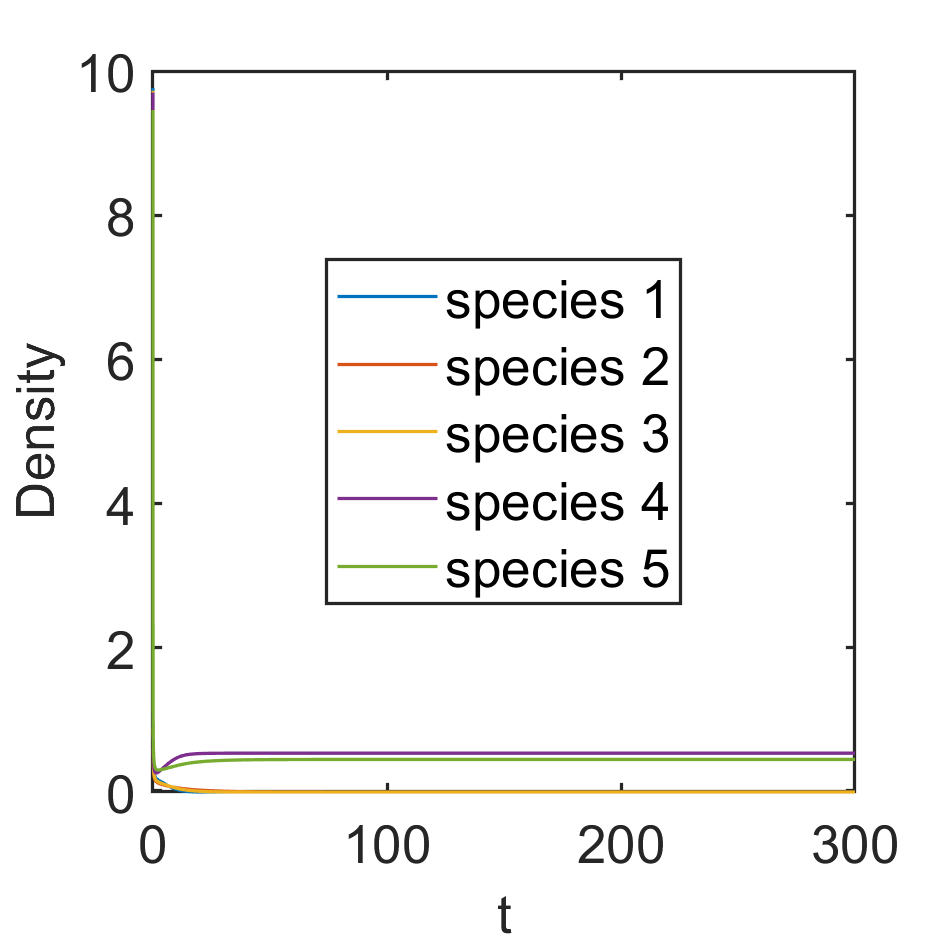}
\captionsetup{width=.95\textwidth}
    \caption{}
    \label{fig:wsabis2}
  \end{subfigure}
\caption{
(a) From a low-level initial condition, species 2 and 5 are winners.
(b) From a high-level initial condition, species 4 and 5 are winners. The system has a bi-stability.
}
\end{figure*}

Finally, we use a numerical example to illustrate Theorem \ref{thm:st}. We consider an equation system $Ax^2+Bx=\mathbf{1}$. 
%We use a simple example. 
For the tensor $A$, its diagonal entries are $A_{111}=A_{222}=11$ and all off-diagonal are all $-1$. For the matrix $B$, $B_{11}=2, B_{12}=-1, B_{21}=-1, B_{22}=2$. Let $v=(1,1)^\top$. We can check that $Av^2>\mathbf{0}$ and $Bv>\mathbf{0}$. Actually, the tensors $A,B$ are all $\mathcal{M}-$, $\mathcal{H}^+-$, $\mathcal{S}-$tensors. Further simulations suggest that $Ax^2+Bx=\mathbf{1}$ has a unique positive solution $(0.27016, 0.27016)^\top.$

Then, we use a more general case. For the tensor $A$, its entries are $A_{111}=A_{222}=11, A_{112}=A_{122}=A_{211}=A_{212}=-1,A_{121}=10,A_{221}=5$. For the matrix $B$, $B_{11}=0.8883, B_{12}=1, B_{21}=0, B_{22}=1$. Let $v=(0.1117,1)^\top$. We can check that $Av^2>\mathbf{0}$ and $Bv>\mathbf{0}$. The tensors $A,B$ are all $\mathcal{S}-$tensors but not $\mathcal{H}^+-$ tensors. Thus, the conditions of Theorem \ref{thm:st} are satisfied. Numerical simulation suggests that $Ax^2+Bx=\mathbf{1}$ has a unique positive solution $(0.49547, 0.43791)^\top.$ 

\section{Conclusion}
This paper studies a general higher-order Lotka-Volterra model. More precisely, we consider 3 typical scenarios: full cooperation, competition between 2 factions, and pure competition respectively. For the analysis tool, we provide an existence result of a positive equilibrium for a non-homogeneous polynomial system and give a lower and upper bound for the solution of a polynomial complementarity problem under mild conditions, which is an extension of the current results in tensor algebras. Then, we further utilize the properties of tensor and monotone system theory to provide the results regarding the existence, uniqueness, and stability of the equilibrium of the system. This paper yields a comprehensive understanding of a higher-order Lotka-Volterra model. Finally, all theoretical results are illustrated by some numerical examples.

\bibliographystyle{IEEEtran}
\bibliography{bib}

\begin{thebibliography}{10}
\providecommand{\url}[1]{#1}
\csname url@samestyle\endcsname
\providecommand{\newblock}{\relax}
\providecommand{\bibinfo}[2]{#2}
\providecommand{\BIBentrySTDinterwordspacing}{\spaceskip=0pt\relax}
\providecommand{\BIBentryALTinterwordstretchfactor}{4}
\providecommand{\BIBentryALTinterwordspacing}{\spaceskip=\fontdimen2\font plus
\BIBentryALTinterwordstretchfactor\fontdimen3\font minus
  \fontdimen4\font\relax}
\providecommand{\BIBforeignlanguage}[2]{{%
\expandafter\ifx\csname l@#1\endcsname\relax
\typeout{** WARNING: IEEEtran.bst: No hyphenation pattern has been}%
\typeout{** loaded for the language `#1'. Using the pattern for}%
\typeout{** the default language instead.}%
\else
\language=\csname l@#1\endcsname
\fi
#2}}
\providecommand{\BIBdecl}{\relax}
\BIBdecl

\bibitem{lotka1920analytical}
A.~J. Lotka, ``Analytical note on certain rhythmic relations in organic
  systems,'' \emph{Proceedings of the National Academy of Sciences}, vol.~6,
  no.~7, pp. 410--415, 1920.

\bibitem{volterra1928variations}
V.~Volterra, ``Variations and fluctuations of the number of individuals in
  animal species living together,'' \emph{ICES Journal of Marine Science},
  vol.~3, no.~1, pp. 3--51, 1928.

\bibitem{goh1976global}
B.~Goh, ``Global stability in two species interactions,'' \emph{Journal of
  Mathematical Biology}, vol.~3, no.~3, pp. 313--318, 1976.

\bibitem{sb2010}
\BIBentryALTinterwordspacing
S.~Baigent, \emph{Lotka-Volterra Dynamics — An Introduction.}\hskip 1em plus
  0.5em minus 0.4em\relax Unpublished Lecture Notes, University of College,
  London, 2010. [Online]. Available:
  \url{http://www.ltcc.ac.uk/media/london-taught-course-centre/documents/Bio-Mathematics-(APPLIED).pdf}
\BIBentrySTDinterwordspacing

\bibitem{goh1979stability}
B.~Goh, ``Stability in models of mutualism,'' \emph{The American Naturalist},
  vol. 113, no.~2, pp. 261--275, 1979.

\bibitem{takeuchi1978stability}
Y.~Takeuchi, N.~Adachi, and H.~Tokumaru, ``The stability of generalized
  volterra equations,'' \emph{Journal of Mathematical Analysis and
  Applications}, vol.~62, no.~3, pp. 453--473, 1978.

\bibitem{takeuchi1996global}
Y.~Takeuchi, \emph{Global dynamical properties of Lotka-Volterra
  systems}.\hskip 1em plus 0.5em minus 0.4em\relax World Scientific, 1996.

\bibitem{FB-LNS}
\BIBentryALTinterwordspacing
F.~Bullo, \emph{Lectures on Network Systems}, {1.6}~ed.\hskip 1em plus 0.5em
  minus 0.4em\relax Kindle Direct Publishing, 2022. [Online]. Available:
  \url{http://motion.me.ucsb.edu/book-lns}
\BIBentrySTDinterwordspacing

\bibitem{abrams1983arguments}
P.~A. Abrams, ``Arguments in favor of higher order interactions,'' \emph{The
  American Naturalist}, vol. 121, no.~6, pp. 887--891, 1983.

\bibitem{mayfield2017higher}
M.~M. Mayfield and D.~B. Stouffer, ``Higher-order interactions capture
  unexplained complexity in diverse communities,'' \emph{Nature ecology \&
  evolution}, vol.~1, no.~3, pp. 1--7, 2017.

\bibitem{letten2019mechanistic}
A.~D. Letten and D.~B. Stouffer, ``The mechanistic basis for higher-order
  interactions and non-additivity in competitive communities,'' \emph{Ecology
  letters}, vol.~22, no.~3, pp. 423--436, 2019.

\bibitem{singh2021higher}
P.~Singh and G.~Baruah, ``Higher order interactions and species coexistence,''
  \emph{Theoretical Ecology}, vol.~14, no.~1, pp. 71--83, 2021.

\bibitem{gibbs2022coexistence}
T.~Gibbs, S.~A. Levin, and J.~M. Levine, ``Coexistence in diverse communities
  with higher-order interactions,'' \emph{Proceedings of the National Academy
  of Sciences}, vol. 119, no.~43, p. e2205063119, 2022.

\bibitem{ding2016solving}
W.~Ding and Y.~Wei, ``Solving multi-linear systems with m-tensors,''
  \emph{Journal of Scientific Computing}, vol.~68, no.~2, pp. 689--715, 2016.

\bibitem{wang2019existence}
X.~Wang, M.~Che, and Y.~Wei, ``Existence and uniqueness of positive solution
  for h+-tensor equations,'' \emph{Applied Mathematics Letters}, vol.~98, pp.
  191--198, 2019.

\bibitem{liu2022further}
L.~Liu, X.~Li, and S.~Liu, ``Further study on existence and uniqueness of
  positive solution for tensor equations,'' \emph{Applied Mathematics Letters},
  vol. 124, p. 107686, 2022.

\bibitem{takeuchi1980existence}
Y.~Takeuchi and N.~Adachi, ``The existence of globally stable equilibria of
  ecosystems of the generalized volterra type,'' \emph{Journal of Mathematical
  Biology}, vol.~10, no.~4, pp. 401--415, 1980.

\bibitem{huang2019tensor}
Z.-H. Huang and L.~Qi, ``Tensor complementarity problems—part i: basic
  theory,'' \emph{Journal of Optimization Theory and Applications}, vol. 183,
  pp. 1--23, 2019.

\bibitem{gowda2016polynomial}
M.~S. Gowda, ``Polynomial complementarity problems,'' \emph{arXiv preprint
  arXiv:1609.05267}, 2016.

\bibitem{hirsch2006monotone}
M.~W. Hirsch and H.~Smith, ``Monotone dynamical systems,'' \emph{Handbook of
  differential equations: ordinary differential equations}, vol.~2, pp.
  239--357, 2006.

\bibitem{smith1988systems}
H.~L. Smith, ``Systems of ordinary differential equations which generate an
  order preserving flow. a survey of results,'' \emph{SIAM review}, vol.~30,
  no.~1, pp. 87--113, 1988.

\bibitem{ye2022convergence}
M.~Ye, B.~D. Anderson, and J.~Liu, ``Convergence and equilibria analysis of a
  networked bivirus epidemic model,'' \emph{SIAM Journal on Control and
  Optimization}, vol.~60, no.~2, pp. S323--S346, 2022.

\bibitem{gracy2022endemic}
S.~Gracy, M.~Ye, B.~Anderson, and C.~A. Uribe, ``On the endemic behavior of a
  competitive tri-virus sis networked model,'' \emph{arXiv preprint
  arXiv:2209.11826}, 2022.

\bibitem{smith1986competing}
H.~L. Smith, ``Competing subcommunities of mutualists and a generalized kamke
  theorem,'' \emph{SIAM Journal on Applied Mathematics}, vol.~46, no.~5, pp.
  856--874, 1986.

\bibitem{kawano2022contraction}
Y.~Kawano and M.~Cao, ``Contraction analysis of virtually positive systems,''
  \emph{Systems \& Control Letters}, vol. 168, p. 105358, 2022.

\bibitem{ding2013m}
W.~Ding, L.~Qi, and Y.~Wei, ``M-tensors and nonsingular m-tensors,''
  \emph{Linear Algebra and Its Applications}, vol. 439, no.~10, pp. 3264--3278,
  2013.

\bibitem{zhang2014m}
L.~Zhang, L.~Qi, and G.~Zhou, ``M-tensors and some applications,'' \emph{SIAM
  Journal on Matrix Analysis and Applications}, vol.~35, no.~2, pp. 437--452,
  2014.

\bibitem{cui2024metzler}
S.~Cui, G.~Zhang, H.~Jard{\'o}n-Kojakhmetov, and M.~Cao, ``On metzler positive
  systems on hypergraphs,'' \emph{arXiv preprint arXiv:2401.03652}, 2024.

\bibitem{qi2005eigenvalues}
L.~Qi, ``Eigenvalues of a real supersymmetric tensor,'' \emph{Journal of
  symbolic computation}, vol.~40, no.~6, pp. 1302--1324, 2005.

\bibitem{xu2019estimations}
Y.~Xu, W.~Gu, and Z.-H. Huang, ``Estimations on upper and lower bounds of
  solutions to a class of tensor complementarity problems,'' \emph{Frontiers of
  Mathematics in China}, vol.~14, pp. 661--671, 2019.

\bibitem{bick2021higher}
C.~Bick, E.~Gross, H.~Harrington, and M.~Schaub, ``What are higher order
  networks?'' \emph{SIAM Review}, 2022.

\bibitem{gallo1993directed}
G.~Gallo, G.~Longo, S.~Pallottino, and S.~Nguyen, ``Directed hypergraphs and
  applications,'' \emph{Discrete applied mathematics}, vol.~42, no. 2-3, pp.
  177--201, 1993.

\bibitem{ye2021applications}
M.~Ye, J.~Liu, B.~D. Anderson, and M.~Cao, ``Applications of the
  poincar{\'e}--hopf theorem: Epidemic models and lotka--volterra systems,''
  \emph{IEEE Transactions on Automatic Control}, vol.~67, no.~4, pp.
  1609--1624, 2021.

\bibitem{goh1978sector}
B.~S. Goh, ``Sector stability of a complex ecosystem model,''
  \emph{Mathematical Biosciences}, vol.~40, no. 1-2, pp. 157--166, 1978.

\bibitem{hofbauer1998evolutionary}
J.~Hofbauer and K.~Sigmund, \emph{Evolutionary games and population
  dynamics}.\hskip 1em plus 0.5em minus 0.4em\relax Cambridge university press,
  1998.

\bibitem{wijeratne2009bifurcation}
A.~Wijeratne, F.~Yi, and J.~Wei, ``Bifurcation analysis in the diffusive
  lotka--volterra system: An application to market economy,'' \emph{Chaos,
  Solitons \& Fractals}, vol.~40, no.~2, pp. 902--911, 2009.

\bibitem{hidayati2021stability}
T.~Hidayati and W.~Kurniawan, ``Stability analysis of lotka-volterra model in
  the case of interaction of local religion and official religion,''
  \emph{International Journal of Educational Research and Social Sciences
  (IJERSC)}, vol.~2, no.~3, pp. 542--546, 2021.

\bibitem{zeeman1995extinction}
M.~L. Zeeman, ``Extinction in competitive lotka-volterra systems,''
  \emph{Proceedings of the American Mathematical Society}, vol. 123, no.~1, pp.
  87--96, 1995.

\end{thebibliography}
\end{document}